\documentclass[a4paper,USenglish]{article}

\usepackage{amsmath, amssymb, amsthm, bm, hyperref, graphicx, listings, booktabs, xcolor, cleveref, thm-restate, subcaption, hyperref, thm-restate}

\usepackage{bbm}
\usepackage{lipsum}
\usepackage{tablefootnote}
\usepackage{multirow}
\usepackage{xurl}

\usepackage[margin=1in]{geometry}
\title{Multi-Currency AMMs for Decentralized FOREX Markets: Feasibility \& Optimal Design}

\author{Reina Ke Xin Li\\
University of Toronto, Canada\\
\href{mailto:reinakx.li@mail.utoronto.ca}{\tt reinakx.li@mail.utoronto.ca}\\
\and
Andreas Park\\
University of Toronto, Canada\\
\href{mailto:andreas.park@rotman.utoronto.ca}{\tt andreas.park@rotman.utoronto.ca}\\
\and
Andreas Veneris\\
University of Toronto, Canada\\
\href{mailto:veneris@eecg.toronto.edu}{\tt veneris@eecg.toronto.edu}\\
\and
Srisht Fateh Singh\\
University of Toronto, Canada\\
\href{mailto:srishtfateh.singh@mail.utoronto.ca}{\tt srishtfateh.singh@mail.utoronto.ca}\\
}

\crefformat{footnote}{#2\footnotemark[#1]#3}
\newtheorem{theorem}{Theorem}
\newtheorem{lemma}[theorem]{Lemma}
\newtheorem{proposition}[theorem]{Proposition}
\newtheorem*{proposition*}{Proposition}
\newtheorem*{lemma*}{Lemma}

\newcommand{\IPL}{\text{IPL}}

\begin{document}

\maketitle

\begin{abstract}
Most currency pairs lack a direct liquid market, so international foreign exchange relies on routing transactions through a dominant vehicle currency. Multi-currency automated market makers (AMMs) offer an alternative by sharing liquidity across many currency pairs, facilitating direct cross-currency trade while exploiting liquidity consolidation. This paper studies a multi-currency pool design that minimizes trading cost. Under a constant-mean AMM architecture, equilibrium trading costs reflect the trade-off between reduced price impact from consolidated liquidity and increased impermanent loss from joint return risk. This work derives closed-form costs, characterizes optimal pool weights, and shows that the optimized multi-currency pool dominates the status quo over a range of market parameters. It then formulates the system-level problem of partitioning currencies into multi-currency pools, which is solved using a hierarchical agglomerative clustering algorithm. Empirically, using exchange rate and trade data for 43 currencies over 2008-2023, the algorithm runs in 1.6 seconds and produces pools with geographic and economic structure. Notably, this reduces realized costs by $\sim$13\% relative to the status quo of vehicle-currency routing, with gains stable through episodes of global financial stress.
\end{abstract}
 
\section{Introduction}
\label{sec:intro}
Daily trading volume in foreign exchange markets more than doubled from \$4.0 trillion in 2010 to \$9.6 trillion in 2025~\cite{bis2010triennial,bis2025triennial}.
Despite this growth, foreign exchange transfers follow a peculiar arrangement: most currency pairs do not trade directly with one another but route through the US\ dollar (USD).  For instance, a firm in Australia importing goods from Sweden 
cannot exchange Australian dollars (AUD) for Swedish kronor (SEK) directly at a 
``tight'' spread as direct AUD/SEK markets are thin. Instead, it must 
first convert the AUD to USD and 
then the USD to SEK, thus paying the transaction costs for both trades.
The work of~\cite{bordier2026dollar} estimates that, in 2022, direct trades between the AUD and the SEK were virtually non-existent. More broadly, 13\% of foreign exchange trading volume involving the USD is generated by routing trades between non-dollar currencies through the USD~\cite{somogyi2026dollar}. This is a significant volume given that USD-involved foreign exchange transactions as a whole account for over 85\% of the global foreign exchange volume.

This dependence on the USD as a \emph{vehicle currency} persists as an emergent feature of the international monetary system. After World War II, the stable US economy with its dollar attached to the gold standard made it an attractive peg codified through the 1944 Bretton Woods Agreement. This allowed the USD to play an important role as a vehicle currency for international foreign exchange. Even after the end of the Bretton Woods era in 1971 with the Nixon shock that canceled the convertibility of USD to gold, the USD's dominance as a vehicle currency has been sustained through various economic forces. Among other reasons, when bilateral trade  between two currencies is moderate or low, neither side can sustain a deep enough market to offer competitive spreads, so traders gravitate toward an intermediate currency with high liquidity. Early work by  \cite{krugman:1984}, in fact, makes the case  that concentrating trade in a vehicle currency is more efficient than trading with multiple pairs. 

Yet, the political landscape is shifting in ways that make the architecture of international payments a pressing policy question. The use of dollar-based payments systems as a political instrument through sanctions and freezing of sovereign reserves has prompted countries to seek alternatives~\cite{fabrichnaya2023what,sharma2024loaded,marrow2024russia,donovan2024axis,wong2024saudi}. New currency swap agreements between central banks have been negotiated and payment systems outside the SWIFT network are under active development~\cite{steil2026central,cash2023china,choukeir2024uae,freidin2024brics,mea2026chairs}. Most importantly, technological advances and innovations in  blockchain-based trading technology raise the question of whether  novel currency exchange arrangements can reduce reliance  
on a vehicle currency but also  overall transaction costs.

This paper posits that Decentralized Finance (DeFi) offers a compelling alternative. Specifically, an \emph{automated market maker} (AMM) is a decentralized exchange protocol that holds pooled reserves of assets that are traded against at prices governed by a mathematical \emph{invariant}.
The Balancer protocol, in particular, offers a \emph{multi-asset AMM}~\cite{balancer}. Such a protocol could be the backbone for direct markets between various currency pairs, while still allowing intermediation through a vehicle currency when necessary. This practice has the potential to reduce overall trading costs relative to the status quo system of routing all trades through a vehicle currency by broadening support for {\em direct} cross pair trading.

The BIS's Project Mariana demonstrates the technical feasibility of AMM-based trading and settlement for foreign exchange trades~\cite{bis2023project}.
However, whether a multi-currency AMM arrangement can achieve economic improvements on the status quo depends on a fundamental trade-off. Liquidity providers, {\em i.e.,}  investors who deposit funds into an AMM pool, are exposed to \emph{impermanent loss}: when exchange rates move, arbitrageurs profit against the pool. This is especially so for pools of heterogeneous assets, each with their own price movements and risk. An AMM pool planner must design fees, exposure weights, and currency inclusion to balance this risk against the liquidity benefits of consolidation. At the same time, pooling currencies with volatile or negatively correlated returns amplifies impermanent loss and raise the fee(s) required to sustain the pool, undermining the gains from deeper liquidity.

Resolving these trade-offs is the central contribution of this paper. \emph{First}, we extend to multi-currency AMMs a mathematical framework in which equilibrium trading costs---the costs under optimally designed pool parameters assuming competitive liquidity provision---are a function of trading volumes, exchange rate covariances, and pool composition. \emph{Second}, we show the market conditions under which an optimally designed multi-currency AMM dominates the status quo and study optimal pool weights in closed form for these conditions. We further develop a tractable approximation for the general case. \emph{Finally}, we examine the system-level problem of which currencies to pool together using a \emph{hierarchical agglomerative clustering} (HAC) algorithm on a correlation-based distance metric. Empirically, applying the framework to 43 non-pegged currencies over 2008--2023, the HAC algorithm runs in 1.6 seconds and produces pools that demonstrate meaningful \emph{regional} and {\it economic} structure. The resulting architecture reduces realized aggregate costs by $\sim$13\%---USD $\sim$2.5 billion---relative to the status quo, with gains remaining stable through  episodes of financial stress.\footnote{The code can be found at \url{https://github.com/anonauthor2/Multi-Currency-AMMs-FOREX}.}

The paper proceeds as follows.
Section~\ref{sec:setup} introduces the general AMM framework. Section~\ref{sec:overview} contextualizes this setup. Section~\ref{sec:three-currency} studies benchmark exchange arrangements. 
Section~\ref{sec:general-amm} analyzes equilibrium costs of multi-currency AMM pools. 
Section~\ref{sec:param} derives analytical weight optimization results for parametrized symmetric currencies. Section~\ref{sec:general-opt} studies the general weight optimization problem and develops a
closed-form approximation.
Section~\ref{sec:pool-selection} formulates the currency partitioning problem,
describes the HAC algorithm, and presents an empirical evaluation.
Section~\ref{sec:related} reviews the literature.
Section~\ref{sec:conclusion} concludes and discusses further directions.

\section{Preliminaries: General AMM Framework}
\label{sec:setup}

Consider an AMM for $N+1$ currencies indexed by $i \in \{0,1,\ldots,N\}$, funded by \emph{liquidity providers} who deposit reserves $R_i \geq 0$ of each currency. The AMM is governed by an \emph{invariant function} of the reserves that implicitly determines the prices at which the AMM trades. When a trader exchanges asset $i$ for asset $j$ by submitting $\Delta_i > 0$ units, the outgoing quantity $\Delta_j$ is uniquely determined by the requirement that the invariant is preserved.

Throughout, we use the \emph{constant mean} invariant introduced by Balancer~\cite{balancer}, $\prod_{i=0}^N R_i^{w_i}=k$, for a fixed weight vector $(w_0,\ldots,w_N)$ with $w_i \geq 0$ and $\sum_i w_i = 1$. The weight $w_i$ governs pool exposure to currency $i$: at fundamental prices, currency $i$ accounts for fraction $w_i$ of total pool value. Thus, the 2-asset constant-product AMM ($R_0 R_1=k$) is the special case used by Uniswap~V2, among other popular AMMs~\cite{uniswapv2}.

\subsection{Prices, Fees, and Trading Costs}
\label{subsec:prices}

The AMM-implied marginal price of asset $i$ in terms of asset $j$ is, for the constant-mean AMM, $p_{ij} = (w_i/w_j)(R_j/R_i)$~\cite{balancer}. Competitive arbitrageurs continuously restore the pool's marginal prices to the exogenous exchange rates $S_i$ (with numeraire $S_0 \equiv 1$), so that $p_{ij} = S_i/S_j$ in equilibrium. For the constant-mean AMM this requires $S_i R_i/w_i = S_j R_j/w_j$ for all pairs.

A finite trade of size $\Delta$ in numeraire moves the reserve vector along the invariant surface, causing the effective price $S^{\text{eff}}(\Delta)$ paid by the trader to deviate from the pre-trade spot price $S$. The \emph{price impact} at spot price $S$ is
\begin{equation}\label{eq:priceimpact}
  \Pi(\Delta) = \frac{S^{\text{eff}}(\Delta) - S}{S},
\end{equation}
which is strictly positive and increasing in $\Delta$ for the constant-mean invariant. Traders additionally pay an explicit fee $f \geq 0$ proportional to transaction value,\footnote{For simplicity, we assume that fees are paid to liquidity providers directly, not added back to the pool.} so the total \emph{unit trading cost} is $c(\Delta,f) = \Pi(\Delta) + f$. Denoting by $c_{ij}$ the unit cost for pair $i,j$ and $Q_{ij}$ the total trading volume on pair $(i,j)$ (assuming a fixed trade size $\Delta$ and balanced flows in either direction), the aggregate cost for pair $(i,j)$ over the period is approximately $c_{ij}(\Delta,f) = \bigl(\Pi(\Delta) + f\bigr)Q_{ij}$.

\subsection{Impermanent Loss}
\label{subsec:ipl}

The pool's value in the numeraire is $V = \sum_i S_i R_i$, held collectively by liquidity providers. 
When external prices change, arbitrageurs trade against the pool to restore price alignment.
Arbitrageurs are incentivized by the favourable terms offered by the AMM: the AMM sells the asset below its new fundamental value or buys above it.
The profit earned by the arbitrageur is at the expense of the pool's liquidity providers.
This gives rise to \emph{impermanent loss} (IPL).
Consider an initial deposit $V(0)$ and a period $[0,T]$ in which currency $i$'s has gross return $x_i = \frac{S_i(T)}{S_i(0)}$ (with $x_0 \equiv 1$). After arbitrageurs restore price alignment, the pool value is $V(T)$, while the counterfactual hold value of the initial deposit is $V_{\text{hold}}(T) = \sum_i x_i S_i(0) R_i(0)$. The IPL is defined as\footnote{IPL is also sometimes defined in other literature as $(V(T)-V_\text{hold}(T))/V_\text{hold}(T)$.}
\begin{equation}\label{eq:ipl}
  \IPL = \frac{V(T) - V_{\text{hold}}(T)}{V(0)} \leq 0.
\end{equation}
The loss is described as ``impermanent'' in DeFi settings because it is unrealized if prices revert before the provider withdraws.
In the context of this work, the loss is due to changes in the fundamental.

\subsection{Equilibrium Liquidity Supply and Optimal Fees}
\label{subsec:equilibrium}

We follow the framework of~\cite{malinova2023learning} and model equilibrium competitive liquidity provision: 
risk-neutral providers deposit at the start of a period and withdraw  at the end, earning fee revenue while bearing impermanent loss.
To simplify the exposition, we assume that there is sufficient exogenous, uninformed trading relative to the order flow of informed arbitrageurs so that fees accrue on noise volume while IPL stems from price changes driven by informed flow.
Competitive liquidity providers contribute capital until they \emph{break-even}:
\begin{equation}\label{eq:breakeven}
  E[|\IPL|] \cdot V(0) = E[Q] \cdot f, \qquad Q = \sum_{i<j} Q_{ij}.
\end{equation}
For a given fee rate $f$ and return and volume distributions, \eqref{eq:breakeven} pins down the equilibrium pool depth $V(0)$ and hence the price impact and unit trading cost $c(\Delta,f)$. The pool designer chooses the  fee that minimizes unit trading costs, $f^* = \arg\min_f c(\Delta,f)$; this fee maximizes trader welfare and ensures that no competing pool can charge better rates. 

\section{Overview}\label{sec:overview}
At a high level, our goal is to design an optimal foreign exchange mechanism and ask: \emph{do multi-currency pools lower costs} or {\it does the additional risk from holding many currencies outweigh the benefits of pooling?} 
Our benchmarks are bilateral AMMs as proxies for the cost of vehicle routing.
Although at first sight these AMMs look different to traditional markets, they are economically equivalent to a limit order book and as \cite{malinova2023learning} show, bilateral AMMs are a low cost benchmark for existing markets. 
We use the framework in Section~\ref{sec:setup} to isolate the key forces determining costs. 
The AMM invariant determines how trades move pool prices, giving rise to price impact---a cost that diminishes with increased liquidity. 
Liquidity providers supply this liquidity, but exchange rate volatility exposes them to impermanent loss, so they require trading fees as compensation. 
A higher fee attracts more liquidity but is charged to the trader. 
A pool planner therefore sets a fee that balances these forces, yielding  optimal equilibrium pool depth and minimized trading cost.

The following builds on this framework by studying how pools should be designed once these forces are in place, moving from simple benchmarks to richer design problems. Constant-mean pools introduce weights as an additional design parameter that control the pool's exposure to each currency's risk and volume. At the system level, design also involves deciding which currencies to pool together. We ultimately study the pool weights and membership structures that minimize equilibrium trading costs across different market environments.

\section{Benchmark Currency Exchange Architectures}
\label{sec:three-currency}
 
We begin by specializing to a minimal three-currency environment consisting of currencies $\{0, 1, 2\}$, with currency $0$ as the numeraire, to establish benchmarks, and then generalize the architectures to $N+1$ currencies. Following the work in~\cite{malinova2023dethroning}, we consider three arrangements. 
In the \emph{status quo}, trades between currencies $1$ and $2$ are routed through the vehicle currency $0$ via two separate bilateral constant-product pools; each $(0,i)$ pool handles total volume $Q_i = \sum_{j \neq i} Q_{ij}$, and a \emph{cross-pair} trader (trading currencies $1$ and $2$) bears costs in both. 
 In the \emph{dedicated bilateral pools} architecture, each pair has its own pool. In the \emph{equal-weighted constant mean pool}, all three currencies are held in a constant mean pool with $w_0 = w_1 = w_2 = 1/3$. By \cite{malinova2023dethroning}, the fee-optimized equilibrium unit costs for trades between $i\leftrightarrow j$ of size $\Delta$ in numeraire under the three arrangements are, respectively,
\[
  c_{ij}^\text{SQ} = \sqrt{\frac{\sigma_i^2\, \Delta}{E[Q_i]}} + \sqrt{\frac{\sigma_j^2\, \Delta}{E[Q_j]}}, \qquad
  c_{ij}^\text{BP} = \sqrt{\frac{\sigma_{ij}^2\, \Delta }{E[Q_{ij}]}}, \qquad c_{ij}^\text{SYM} = \sqrt{\frac{2H \Delta}{3 E[Q]}}
\]
where $\sigma_i^2 = \mathrm{Var}(x_i)$, $\sigma_{ij}^2 = \mathrm{Var}(x_i - x_j)$ (we call $\sigma$ the relative volatility of $i$ and $j$), $H = \sum_{i<j} \sigma_{ij}^2$ and $Q = \sum_{i<j} Q_{ij}$.\footnote{\cite{malinova2023dethroning} derives the costs in expectation using second-order approximations for IPL under the assumption of small, zero-mean excess returns $\epsilon_i = x_i - 1 \approx 0$. We use the same simplifying assumption throughout.}

In each case, costs scale with volatility and shrink with volume. 
The status quo charges two  costs for a $i\leftrightarrow j$
trade, each discounted by its own currency's total volume.
The bilateral pool charges a single cost but only gets liquidity from the $i,j$ pair.
In the equal-weighted AMM all pairs share the same cost discounted by aggregate volume $E[Q]$, so more actively traded pairs effectively subsidize thinner-traded ones.

The cost $c^\text{BP}$ generalizes straightforwardly to $N+1$ currencies, with aggregate cost $c^\text{BP} = \sum_{i<j} \sqrt{\sigma_{ij}^2 E[Q_{ij}]\Delta}$. We extend the results for $c^\text{SQ}$ and $c^\text{SYM}$ to $N+1$ currencies in the following propositions, both
of which are proved in Appendix~\ref{app:benchmarks}.
\begin{proposition}[restate=statusquocost,name=Status quo trading cost]\label{prop:sqcost}
  The expected aggregate trading cost in the status quo for $N+1$ currencies routed through the vehicle currency $0$ is
\[
c^\text{SQ} = \sum_i \sqrt{\sigma_i^2 E[Q_i] \Delta}
\]
\end{proposition}
\begin{proposition}[restate=equalweightscost,name=Equal weights trading cost]\label{prop:symcost}
  The expected aggregate trading cost in an equal-weighted $N+1$ currency constant-mean pool is 
  \[
  c^\text{SYM} = \sqrt{\frac{2 \Delta H E[Q]}{N+1}}
  \]
\end{proposition}
In each case, the aggregate costs scale with exchange rate volatilities, but $c^\text{SQ}$ charges volatility per currency, as each currency routes trades through the vehicle separately, while $c^\text{BP}$ and $c^\text{SYM}$ charge per pair, reflecting the direct cross-pair trading. The dominant pool depends on the specific volatility and volume structure, which is explored further in Section~\ref{sec:param}.

\section{Equilibrium in Constant-Mean AMMs}
\label{sec:general-amm}
Equal weights are not optimal in general: a currency with high individual volatility inflicts more impermanent loss per weight, while one with high trading volume generates more fee revenue per weight. Introducing non-uniform weights allows the pool designer to tilt exposure away from high-volatility, low-volume currencies, reducing equilibrium trading costs.

The break-even framework of Section~\ref{subsec:equilibrium}, developed for two-asset constant-product pools in \cite{malinova2023learning}, extends to general constant-mean pools. The lemmas below derive the corresponding price impact and impermanent loss expressions; applied to the break-even condition~\eqref{eq:breakeven} and cost function $c(\Delta,f)$, they yield the equilibrium trading cost in Proposition~\ref{prop:cost}. All proofs for this section are found in Appendix~\ref{app:ammcosts}.

\begin{lemma}[restate=priceimpact,name=Price impact]\label{prop:priceimpact}
  In the general constant-mean AMM, the price impact of a trade of size
  $\Delta$ in numeraire between currencies $i$ and $j$ is approximately
  \[
    \Pi(\Delta) \approx \frac{w_i + w_j}{w_i w_j} \cdot \frac{\Delta}{2V}
  \]
\end{lemma}
\begin{lemma}[restate=impermanentloss,name=Impermanent loss]\label{prop:ipl}
  For gross returns $x_i$ with $x_i \approx 1$ and $E[x_i] = 1$, the expected impermanent loss of the general constant-mean pool satisfies
  \[
    E[|\IPL|] \approx \frac{H_w}{2},
    \qquad H_w = \sum_{i < j} w_i w_j \sigma_{ij}^2
  \]
\end{lemma}
Each currency pair $(i,j)$
contributes to $H_w$ in proportion to the product of the two currencies'
weights and the variance of their relative returns.  Concentrating weight on a
pair with high relative volatility $\sigma_{ij}^2$ increases expected impermanent
loss and therefore the fee needed to sustain liquidity provision in equilibrium. Together with Lemma~\ref{prop:priceimpact}, this gives us the following proposition.
 
\begin{proposition}[restate=cost,name=Equilibrium trading cost]\label{prop:cost}
  Under the break-even condition~\eqref{eq:breakeven} and
  Lemmas~\ref{prop:priceimpact}--\ref{prop:ipl}, the aggregate equilibrium
  trading cost is minimized at
  \begin{equation}\label{eq:aggregatecost}
    f^* = \sqrt{\frac{\Delta H_w}{4 (E[Q])^2}
      \sum_i \frac{E[Q_i]}{w_i}},
    \qquad
    c^* = \sqrt{\Delta H_w \sum_i \frac{E[Q_i]}{w_i}}
  \end{equation}
\end{proposition}
 
The cost factors cleanly into two terms: $H_w$, which is proportional to expected impermanent
loss, and $\sum_i E[Q_i]/w_i$, which is proportional to the volume-weighted price impact per unit of volume.
The pool designer's goal is to choose weights to minimize this cost.

\section{Optimal Pools Under Symmetric Multi-Currency Parametrizations}
\label{sec:param}
We begin by deriving closed-form optimal weights and cost comparisons for special symmetric cases, where the general problem of Section~\ref{sec:general-amm} reduces to a scalar optimization.
 The  three-currency case extends the analysis of  \cite{malinova2023dethroning} by characterizing optimal weights, while the $N+1$-currency case is, to our knowledge, new.
 
\subsection{Three Currency Case}
\label{subsec:symmetric-param}
 
We first study a parametrization that captures the structure of a
three-currency system organized around a vehicle currency. Let currency $0$ be
the vehicle.  Suppose that the $(0,1)$ and $(0,2)$ pairs have equal trading
volumes and equal relative volatilities, while the $(1,2)$ cross pair trades at a
smaller volume with higher relative volatility. That is, let
$E[Q_{01}] = E[Q_{02}] = q$ and
$\sigma_{01}^2 = \sigma_{02}^2 = \sigma^2$.
Write $E[Q_{12}] = vq$ for some $v>0$ and 
$\sigma_{12}^2 = s^2 \sigma^2$ for $s \geq 0$.\footnote{Since $\sigma_{12}^2 = \mathrm{Var}(\epsilon_1 - \epsilon_2)
= \sigma_{01}^2 + \sigma_{02}^2 - 2\,\mathrm{Cov}(\epsilon_1, \epsilon_2)$,
the Cauchy--Schwarz inequality implies $0 \leq \sigma_{12}^2 \leq 4\sigma^2$ and $s \in [0, 2]$.}
Values $s > \sqrt{2}$ correspond to negatively correlated returns
for currencies $1$ and $2$, while $s < \sqrt{2}$ corresponds
to positively correlated returns.  The case $v < 1$ represents the typical
vehicle-currency scenario in which direct cross-currency trading is thin;
$v > 1$ could represent two highly liquid currencies and one smaller
currency trading against the two.

Under this parametrization, currencies $1$ and $2$ are symmetric, so the
cost-minimizing weights satisfy $w_1 = w_2 = (1-w_0)/2$, and the problem
reduces to optimizing over the scalar $w_0 \in (0,1)$. The result is summarized in the following proposition which is proved in Appendix~\ref{app:param}.

\begin{proposition}[restate=sym,name={Optimal weight and cost for vehicle and 2 symmetric currencies}]\label{prop:sym}
  Under the symmetric $(v,s)$ parametrization for 3 currencies, if $s^2 < 2(1+2v)/(1+v)$, the expected aggregate trading cost
  of Equation~\eqref{eq:aggregatecost} has unique interior minimizer at $w_1^*=w_2^* = \frac{1-w_0^*}{2}$ and
  \[
    w_0^* = \frac{s}{\sqrt{(1+2v)(4-s^2)}},
    \qquad
    c^* = \sqrt{q \sigma^2 \Delta}
      \sqrt{s^2 v + 2 + s\sqrt{(4-s^2)(1+2v)}}
  \]
   If $s^2 \geq 2(1+2v)/(1+v)$, no interior minimizer exists; the infimum over $w_1,w_2>0$ is obtained as $w_0 \rightarrow 1$.
\end{proposition} 
The optimal vehicle-currency
weight $w_0^*$ increases with $s$: when the $(1,2)$ pair is volatile, the pool tilts toward the vehicle
to lower IPL generated by
the cross pair.  Conversely, $w_0^*$ decreases with $v$ (the relative volume of the
non-vehicle pair): when the $(1,2)$ pair has high volume, the pool should provide more depth to the cross pair and reduce its price impact.

We plot the regions in which $c^*$, $c^\text{SYM}$, and $c^\text{BP}$ are optimal over the status quo cost, $c^\text{SQ}$ in Figure~\ref{fig:param_regions}.
The precise expressions for the thresholds are found in Appendix~\ref{app:thresholds}.

\begin{figure}[t]
    \centering
    \begin{minipage}[t]{0.49\textwidth}
        \centering
        \includegraphics[width=\linewidth]{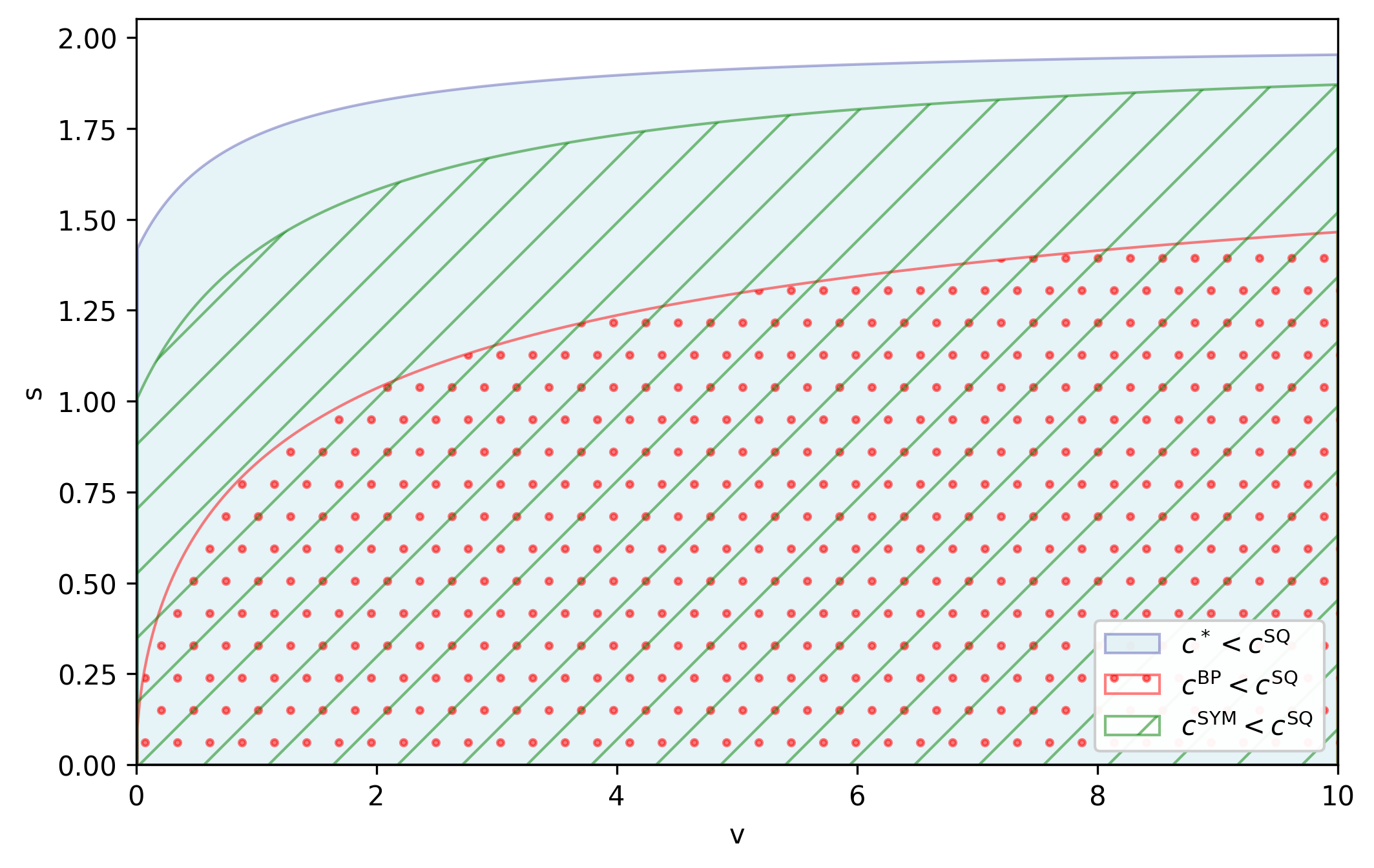}
      \caption{$s$ vs. $v$ regions for optimality over the status quo under the three-currency $(v,s)$ parametrization.}
      \label{fig:param_regions}
    \end{minipage}%
    \hfill%
    \begin{minipage}[t]{0.49\textwidth}
        \centering
        \includegraphics[width=\linewidth]{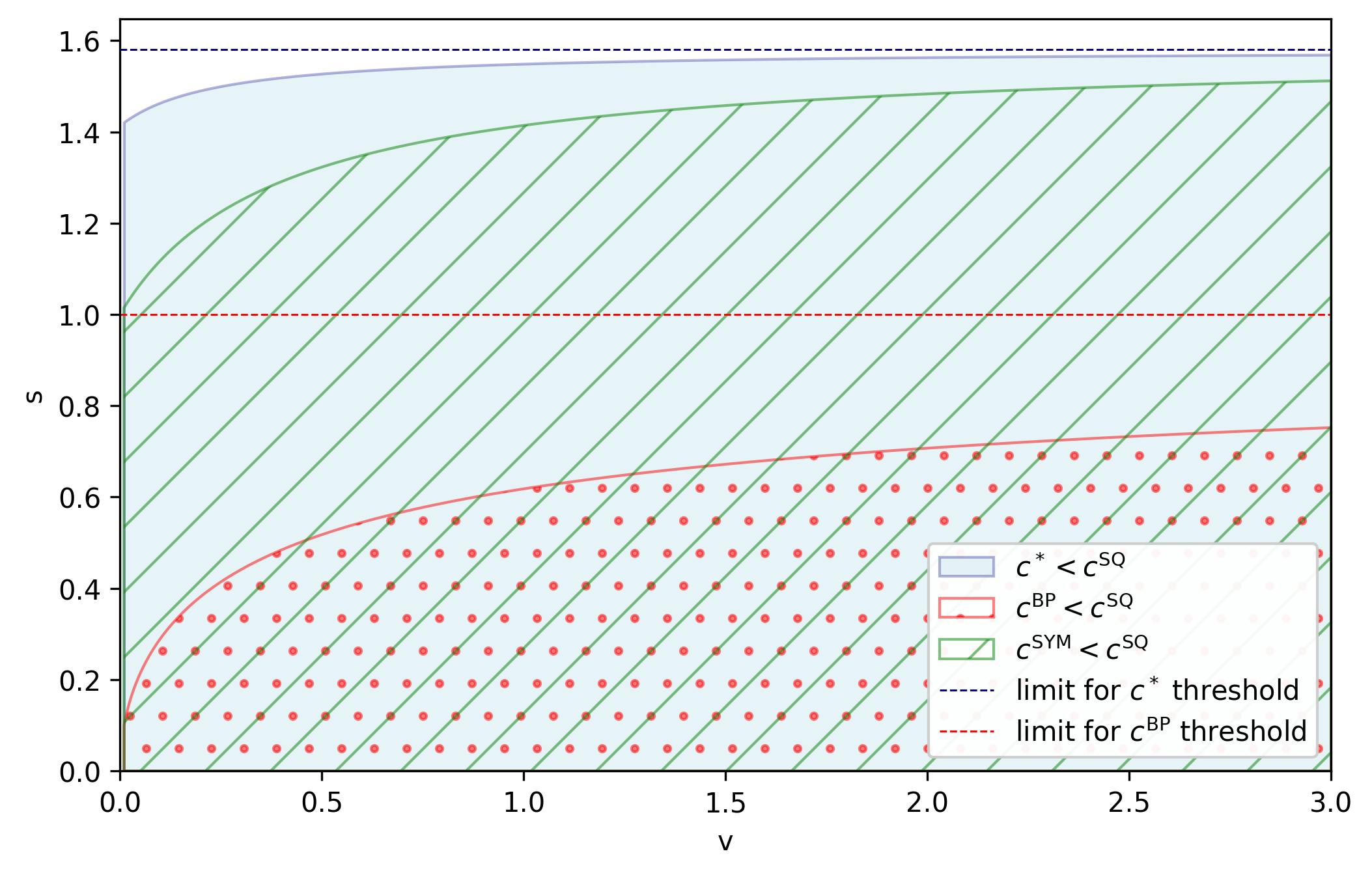}
        \caption{$s$ vs. $v$ regions for optimality over the status quo under the $(v,s)$ parametrization for $N=5$ currencies.}
        \label{fig:param_regions_N}
    \end{minipage}%
    \hfill%
\end{figure}

\subsection{Extension to More Currencies}\label{subsec:hubnspoke}
The above can be generalized to $N\geq 2$ symmetric currencies organized around a vehicle currency.
Again, let currency $0$ be the vehicle and suppose that $(0,i)$ pairs have equal trading
volumes and volatilities, $q$ and $\sigma$, while each pair $(i,j)$ with $i,j \neq 0$ trading at
volume $vq$ with volatility of returns $s\sigma$.
Under this parametrization, the weights satisfy $w_i = (1-w_0)/N$, and the problem reduces to optimizing
over $w_0 \in [0,1]$, summarized in the following proposition and proved in Appendix~\ref{app:param}.
\begin{proposition}[restate=nsym,name={Optimum for vehicle and $N$ symmetric currencies}]\label{prop:nsym}
  Under the symmetric $(v,s)$ parametrization for $N+1$ currencies, if $s^2 <  2(1+Nv)/(1+(N-1)v)$, the expected aggregate trading cost from Equation~\eqref{eq:aggregatecost} has unique interior minimizer at $w_i^* = (1-w_0^*)/N$ for all $i>0$ and
  \[
    w_0^* = \frac{s}{\sqrt{(1+Nv)(2N-(N-1)s^2)}}
  \]
  If $s^2 \geq 2(1+Nv)/(1+(N-1)v)$, no interior minimizer exists; the infimum over $w_1,w_2>0$ is obtained as $w_0 \rightarrow 1$.
\end{proposition}
The expression for the optimal aggregate expected cost $c^*$ can be found as Equation~\eqref{eq:paramcost} in Appendix~\ref{app:param}.
The optimal weights of the three-currency case are recovered by setting $N=2$, and similarly, we observe the increasing optimal vehicle weight with $s$.

We plot, for the $N+1$ currencies case, the regions for which $c^*$, $c^\text{SYM}$, and $c^\text{BP}$ are optimal over the status quo cost, $c^\text{SQ}$ in Figures~\ref{fig:param_regions_N}. The precise expressions for the thresholds are found in the Appendix~\ref{app:thresholds}.

\section{General Constant Mean AMM Weight Optimization}
\label{sec:general-opt}

For a general distribution of volumes and volatilities of a given pool of currencies, the pool designer solves
for the weights $w_i>0$ summing to $1$ and minimizing
$\mathcal{F}(w) = H_w \sum_i E[Q_i]/w_i$. Closed-form solutions are not available in general, but the system can be solved numerically via constrained optimization methods.

A tractable heuristic approximation is obtained by noting that, in the objective
$\mathcal{F}(w) = H_w \cdot \sum_i E[Q_i]/w_i$, the first factor captures
the impermanent loss exposure of the pool and the second captures the
volume-weighted price impact.  The two factors interact through the weights, 
with the first biasing towards currencies with small relative volatilities
and the second biasing towards currencies with high trading volumes.
A useful simplification captures these forces and decouples the weights:
we approximate the minimization objective by the heuristic
$\tilde{\mathcal{F}}(w) = \sum_{i} \frac{E[Q_i]}{w_i H_i}$ 
where $H_i = \sum_{j \neq i} \sigma_{ij}^2$. The heuristic-approximated optimum is presented in the following proposition, proved in Appendix~\ref{app:approx}.
\begin{proposition}[restate=approxi,name=Approximate optimal weights]\label{prop:approx}
  When $H_i>0$ for all $i$ and $E[Q_i]>0$ for some $i$, the minimizer of $\tilde{\mathcal{F}}(w) = \sum_{i} \frac{E[Q_i]}{w_i H_i}$ subject to
  $\sum_{i} w_i = 1$, $w_i > 0$ $\forall i$, is
  \[
    w_i^* = \frac{1}{D} \sqrt{\frac{E[Q_i]}{H_i}},
    \qquad \forall\, i=0,\ldots,N,
    \qquad \text{where } D = \sum_j \sqrt{E[Q_j]/H_j},
  \]
\end{proposition}
This heuristic rule assigns each currency a weight proportional to the square root of its expected aggregate trading
volume $E[Q_i]$ and the inverse of the sum of its relative variances. 
The square-root scaling arises from the diminishing effects of concentrating weight
on a single currency in reducing price impact while leaving the impermanent loss exposure
unimproved. Hence, the optimum disperses weight broadly, with the degree of
dispersion governed by the square-root of relative volumes and volatilities.

\subsection{Numerical Validation}
\label{subsec:numerical}

We evaluate the quality of the approximation $\min \tilde{\mathcal{F}}$ by comparing the
equilibrium trading cost it achieves against the cost attained by the
numerical minimizer of $\mathcal{F}$, across a wide range of pool sizes and
randomly drawn market environments. 

\textbf{Methodology}: For each pool of $N \in \{2,  3, \ldots, 20\}$ non-vehicle currencies (pooled together with the vehicle), we conduct $500$ independently
drawn trials, proceeding as follows. 
To generate covariance matrices, we use a synthetic $4$-factor model~\cite{chamberlain1983arbitrage}, where returns are correlated through shared factors, with loadings $B \in \mathbb{R}^{N \times 4}$, a matrix of standard normal entries.
The covariance matrix $K$ is constructed as $K = BB^T + \Psi$, where $\Psi$ is a diagonal matrix with entries uniformly chosen from $[0,1]$, representing idiosyncratic risk.
An additional column and row of zeros is added for the 0-return vehicle currency.
The pairwise relative variance matrix
$\sigma^2 \in \mathbb{R}^{(N+1) \times (N+1)}$ is then computed with entries $ \sigma_{ij}^2 =
K_{ii} + K_{jj} - 2K_{ij}$.  
Expected trading volumes are drawn as $E[Q_{ij}] = q_i q_j$ with $q_i \sim \text{LogNormal}(0, 1)$,
yielding a symmetric volume matrix with log-normally distributed marginals.
$E[Q_i]$ for all $i$ and $E[Q]$ are
computed from this matrix.
We compute three weight vectors: $w^\mathcal{F}$, the numerical minimizer of $\mathcal{F}(w) = H_w \sum_i E[Q_i]/w_i$ using sequential quadratic programming 
(SLSQP)~\cite{kraft1988software};
$w^{\tilde{\mathcal{F}}}_i = D^{-1}\sqrt{E[Q_i]/H_i}$, the approximation from Proposition~\ref{prop:approx}; and $w_i^\text{EQ} = 1/(N+1)$, the equal weights vector. The aggregate expected cost of Equation~\eqref{eq:aggregatecost} is evaluated for each weight vector with trade size $\Delta = 1$. We also compute the costs of the status quo vehicle routing and the dedicated bilateral pools architectures for comparison. 

\textbf{Results}: Figures~\ref{fig:approx-cost} and~\ref{fig:approx-cv} report the mean aggregate expected cost, expressed as a percentage of total volume, and its coefficient of variation over the trials as a function of pool size $N$. The two methods achieve nearly identical costs for all $N$, with the gap remaining small and roughly constant as $N$ grows from $3$ to $20$. This implies that ignoring the interaction between $H_w$ and $\sum_i E[Q_i]/w_i$ when optimizing weights is negligible in practice. The coefficient of variation tracks closely across both methods for all $N$, indicating that variability reflects heterogeneity of the market environments rather than any degradation of the approximation.
A pool designer can thus use the closed-form rule $w_i^* \propto \sqrt{E[Q_i]/H_i}$ in place of a much more computationally demanding numerical optimizer. This is practical for AMM weight optimization in computationally restricted applications such as in blockchains, or for efficiently estimating costs in the optimization algorithm described in the next section.

\begin{figure}[t]
    \centering
    \begin{minipage}[t]{0.49\textwidth}
        \includegraphics[width=\linewidth]{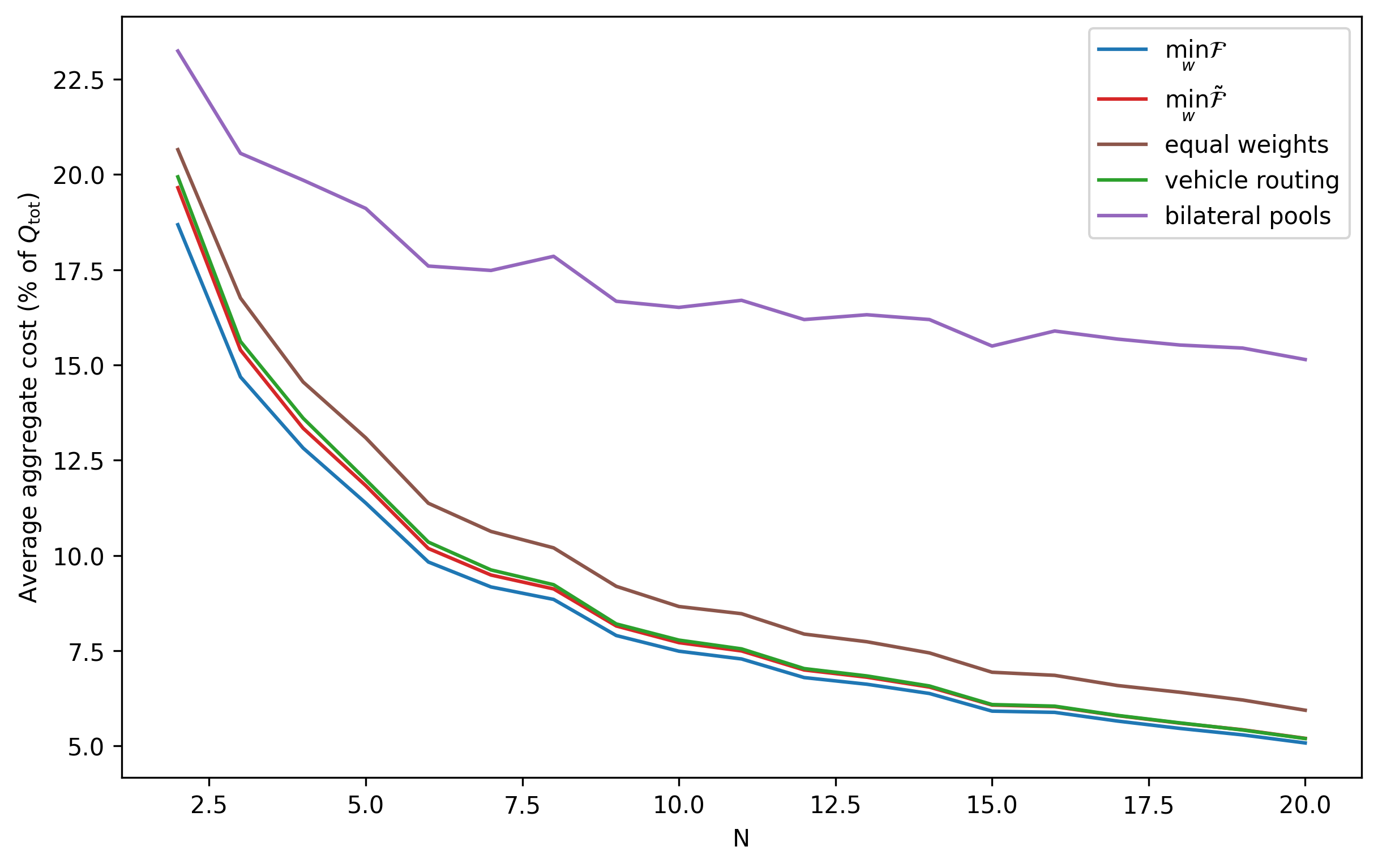}
        \caption{Mean equilibrium aggregate trading cost (as a percentage of total
    volume) achieved by the exact minimizer of $\mathcal{F}$ and the
    closed-form approximate minimizer of $\tilde{\mathcal{F}}$, averaged over $500$
    trials per pool size.}
        \label{fig:approx-cost}
    \end{minipage}
    \hfill
    \begin{minipage}[t]{0.49\textwidth}
        \includegraphics[width=\linewidth]{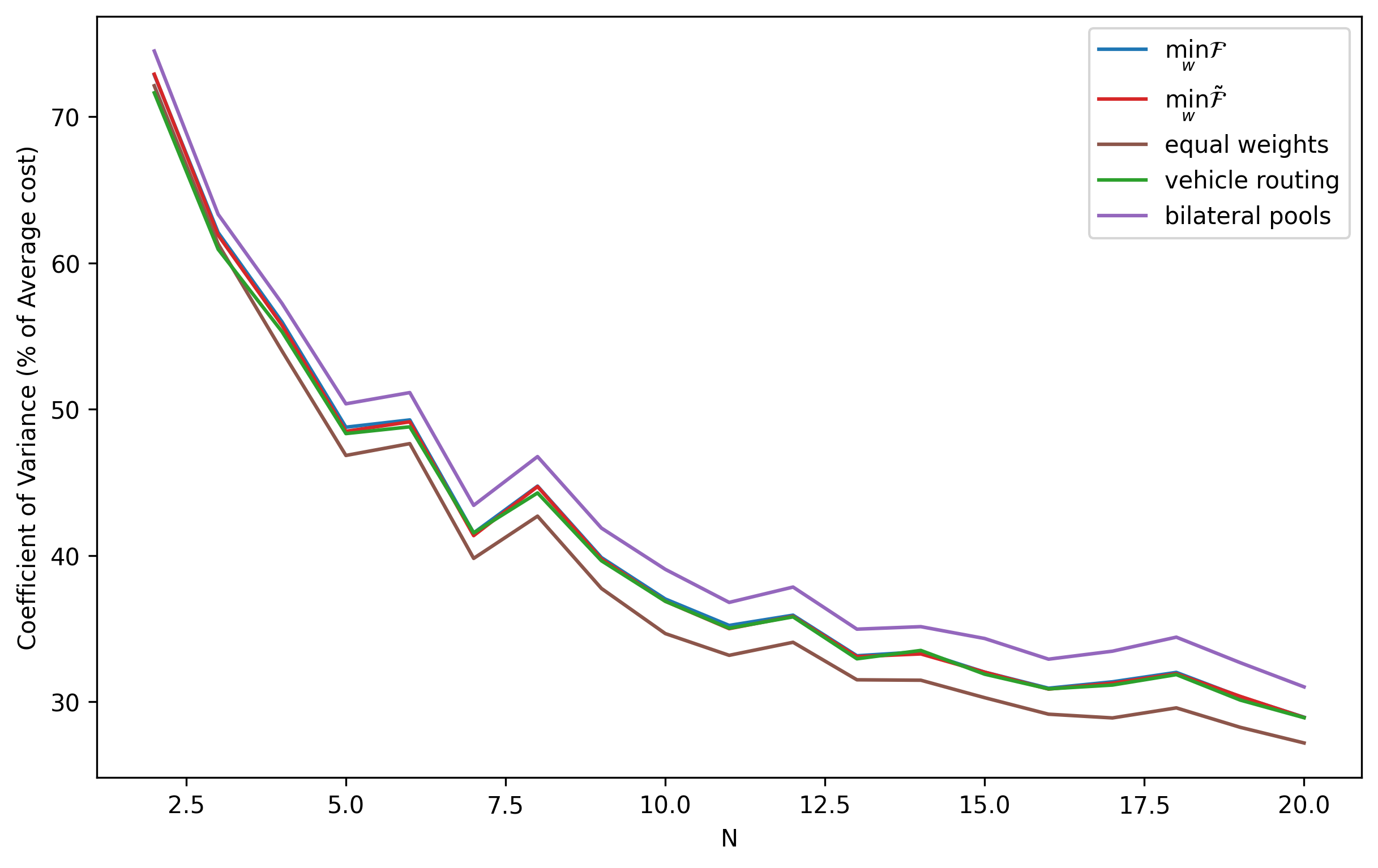}
        \caption{Coefficient of variation (standard deviation / mean) of the
    equilibrium aggregate cost across trials, as a function of pool size $N$,
    for the exact and approximate minimizers.}
        \label{fig:approx-cv}
    \end{minipage}
\end{figure}

\section{Currency Partitioning for System-Level Pool Design}
\label{sec:pool-selection}

In this section, we address the question of \emph{which} currencies to pool together. Evidently, not every currency benefits from co-pooling: if the returns of two currencies have high relative variance $\sigma_{ij}^2$, their impermanent loss raises the fee, which may outweigh liquidity gains. We therefore introduce the \emph{currency partitioning problem}: given $N$ currencies and their pairwise volumes and return covariance structure, partition the currencies into multi-currency pools and unpooled singletons to minimize aggregate trading costs.
 
We fix currency $0$ as the vehicle and consider multi-currency pools of the form $I = \{0\} \cup G$ where $G \subseteq \{1,\ldots,N\}$. 
Let $M = \{I_1,\ldots,I_K,O\}$ be a collection of such pools (that are disjoint except for the vehicle) and the set $O$ of unpooled currencies.
A trade between a non-vehicle $i$ in a pool $I$ and $j \notin I$ is executed in two legs: $i \leftrightarrow 0$ in the pool, and $0 \leftrightarrow j$ in the pool containing $j$ or a bilateral pool with the vehicle if $j \in O$. Trades between unpooled currencies are also routed bilaterally through the vehicle.

Let $c(M)$ denote the total expected system cost; the currency partitioning problem is to choose $M$ minimizing $c(M)$. The expression for $c(M)$ is calculated in Appendix~\ref{app:systemcost}.

\subsection{Rolling Window Realized Cost}
\label{subsec:realized-cost}
 
Our empirical evaluation separates the \emph{ex-ante} and \emph{ex-post} quantities.
Liquidity providers set the pool depth $V$ and fee $f$ before observing the period's realized trading volume and price movements.  
They base these decisions on expectations formed from recent data.

We implement this separation as follows. For each quarter $t$, we form
the \emph{expected} covariance matrix using a 12-month rolling window of monthly returns ending at the start of quarter $t$, scaled
to a quarterly frequency. We form the \emph{expected} volume matrices using the same 12-month windows, scaled to quarterly averages.
The pool planner uses these expected values
to compute the break-even fee $f^*$ via Proposition~\ref{prop:cost} and the optimal weights $w^*$ numerically by minimizing $\mathcal{F}(w)$ at the beginning of each quarter $t$.
Liquidity providers then decide $V = V^*$, the pool depth.

The \emph{realized} cost in quarter $t$ is then computed by applying the unit cost implied by the
ex-ante pool parameters $(w^*, f^*, V^*)$ to the \emph{realized} trade volume for quarter $t$, which is then summed for every quarter over the evaluation period. The expression for the realized cost is also found in Appendix~\ref{app:systemcost}. 
This specification captures an additional risk faced by liquidity 
providers: the fee, weights, and depth are fixed at the start
of the quarter, but the realized cost depends on how much volume actually
flows through the pool. 

\subsection{Hierarchical Agglomerative Clustering Algorithm}
The number of partitions of $N$ items grows super-exponentially, making exhaustive search infeasible. Instead, we exploit the cost function structure to inform heuristic methods. Since a currency inside a multi-asset pool executes all its trades inside the pool, the main determinant of pooling efficiency is the pairwise relative variance of currencies within a pool. Currencies with highly correlated returns have lower $\sigma_{ij}^2$ and thus generate lower impermanent loss. Hence, we use a \emph{hierarchical agglomerative clustering} (HAC) algorithm on the correlations~\cite{manning2008introduction}.

First, we transform correlations into distances $d_{ij} = \sqrt{0.5(1-\rho_{ij})} \in [0,1]$ so that highly correlated currencies are close together. We then construct a partition over the set of currencies as follows. Starting from singletons (equivalent to the status quo), we iteratively merge the two clusters with smallest distance. Here, the distance between two clusters $A$ and $B$ is defined as $d_{AB} = \frac{1}{|A||B|} \sum_{i \in A, j \in B} d_{ij}$, the mean pairwise distance between the currencies across the two clusters. We require $d_{AB} \leq \tau$, a distance threshold.
If all clusters are farther than $\tau$ from each other, the procedure ends and returns the current partition.

We evaluate the partition produced by a range of thresholds $\tau \in [0,1]$ against the expected aggregate cost $c(M)$ (evaluated using approximate weights from Proposition~\ref{prop:approx} for the sake of efficiency) and select the threshold $\tau$ and resulting partition that minimizes this cost.

\subsection{Empirical Data}
\label{subsec:data}

Exchange rates come from the IMF Exchange Rates dataset (monthly period-average \emph{US Dollar per domestic currency} indicator)~\cite{imf_er}, from which we compute monthly returns and relative variances.
Trade volumes are proxied by the USD value (in millions) of bilateral goods exports between the countries issuing the currency. This approximates the real demand for cross-border currency exchange.
This data is obtained from the IMF International Trade in Goods dataset (monthly \emph{Export of Goods, FOB, US dollar} indicator)~\cite{imf_imts}.\footnote{We use the Euro Area (Eurozone) in lieu of a country for the Euro.}  Currency partitioning with the HAC algorithm uses 2002--2007 monthly exchange rates and trade data to calculate correlations, relative variances, and trade volumes over the whole period (scaled to a quarterly basis); out-of-sample 12-month quarterly rolling window relative variances and trade volumes are calculated for the evaluation period, 2008--2023. We choose these windows for several reasons: first, 2002 is the starting year as this is when the Euro was officially adopted; second, we avoid using data from the 2008 financial crisis in the pool selection, but do not avoid evaluating on this crisis period; third, 2023 is the end year as it is the last year for which IMF AREAER exchange rate regime classifications are available.

The last criteria is important as we restrict to non-pegged currencies under the IMF AREAER~\cite{imf_areaer}---currencies classified as independently floating, managed floating, crawling band, crawling peg
(pre-2008 classifications), other managed, free floating, floating, 
crawl-like, crawling peg, or stabilized (post-2008 classifications) over 2002--2023. Pegged currencies are excluded because their near-zero volatility against the peg makes them misleading.
For example, pooling USD/EUR/DKK/BGN, where the 
Danish krone (DKK) and the Bulgarian lev (BGN) are 
pegged to the Euro, results in a realized cost of 
$c(I) =$ 840 USD millions, versus $c^\text{SQ}(I) =$ 1{,}400 USD millions, 
a 40\% cost saving over 2008--2023.
However, this is a flawed analysis---central banks managing pegged currencies tend to provide guaranteed exchange rates against the peg, rather than using market supply and demand.
We further exclude currencies with observed zero returns against the USD for three consecutive months, or that are described as pegged by the IMF AREAER,\footnote{Afghani afghani, Cambodian riel, Dominican peso, Gambian dalasi, Georgian lari, Kyrgyz som, Laotian kip, Liberian dollar, Malagasy ariary, Moldovan leu, Myanmar kyat, and Papua New Guinean kina.} as well as the Haitian gourde (co-circulates with USD\footnote{IBAN lists USD as a currency of Haiti.}) and the Somali shilling (limited regime data). The resulting 43 currencies are listed in Table~\ref{tab:currencies}, with the USD as vehicle (currency $0$). For the remainder of this section, we fix $\Delta = 1$ USD million throughout.
 
\begin{table}[t]
  \centering
  \caption{Currencies with non-pegged exchange rates for all years in 2002--2023 according to IMF AREAER classification (43 currencies).}
  \label{tab:currencies}
  \begin{tabular}{p{0.97\textwidth}}
    \toprule
    Currency code \\
    \midrule
    ALL, AMD, AUD, BIF, BRL, CAD, CDF, CHF, CLP, COP, CZK, DZD, EUR, GBP, GHS, GTQ, IDR, INR, ISK, JMD, JPY, KES, KRW, MUR, MXN, MZN, NIO, NOK, NZD, PEN, PHP, PLN, PYG, RSD, SEK, SGD, THB, TRY, TZS, UGX, UYU, ZAR, ZMW\\
    \bottomrule
  \end{tabular}
\end{table}

\subsection{Empirical Results}
\label{subsec:pool-selection-empirical_hac}
The HAC algorithm, implemented in Python 3.13 with Numpy 2.2.1, Pandas 2.2.3, and SciPy 1.14.1 runs in 1.6 seconds on a 16-core Intel Core Ultra 7 155H processor and 16 GB RAM.
Table~\ref{tab:selected-pools_hac} shows the pool architecture selected by the HAC algorithm, and Table~\ref{tab:realized_costs_hac} reports out-of-sample realized aggregate costs over 2008--2023.
The optimized architecture, with distance threshold $\tau=0.46$, places 18 of the 43 currencies into four multi-asset pools, reducing realized aggregate system costs by approximately 13\% (roughly USD 2.5 billion) relative to bilateral USD routing. Some quarters exhibit realized savings exceeding 15-20\%, and performance persists and remains stable even through the 2008 financial crisis and the 2020 COVID-19 recession, indicating that the selected pool structures are persistent rather than artifacts of the calibration window.\footnote{The HAC algorithm is not theoretically cost-optimizing, but we compare it with a simple greedy algorithm in Appendix~\ref{app:greedy} to benchmark the HAC algorithm's performance.}

Patterns emerge from the resulting pools. 
Pool 1 includes currencies with high foreign exchange turnover often considered `safe-havens' (EUR, JPY, GBP, CHF) and other European currencies, reflecting the institutional and trade integration of the European economic area and its close partners~\cite{bis2025triennial,sato2024time}. Pool 2 groups commodity-linked AUD and NZD with ISK, all three of which are carry-trade target currencies~\cite{chen2003commodity,coudert2014looking}. Pool 3 pairs the Colombian peso and Turkish lira, both emerging markets exposed to high political risks~\cite{vizcaino2026resurgent}.
Pool 4 clusters KRW, SGD, and THB, East and South-East Asian currencies of manufacturing-heavy economies and regional trade integration~\cite{UNComtrade2025}.
This confirms that the trade-off at the heart of this paper---liquidity consolidation versus impermanent loss---can be resolved in practice by grouping currencies that co-move, capturing the benefits of shared liquidity without amplifying exchange rate risk. The emergent geographical and economic structures are a direct consequence of this co-movement pooling. Figure~\ref{fig:correlations_heatmap_hac} visualizes the resulting pools in a heatmap of pairwise correlations, ordered by pool assignment. The heatmap shows the three high-correlation blocks along the diagonal corresponding the Pools 1-4, with the large Pool 1 block most pronounced on the top-right.
Overall, the structured results and 13\% empirical cost reduction serve to validate that well-designed multi-currency pools can meaningfully lower foreign exchange trading costs relative to bilateral vehicle-currency routing.

\begin{table}[t]
\centering
\caption{Pool architecture from HAC algorithm. USD included implicitly in all pools.}
\label{tab:selected-pools_hac}
\begin{tabular}{lp{0.6\linewidth}l}
\toprule
Pool & Currencies & Description \\
\midrule
1 & ALL, CHF, CZK, EUR, GBP, JPY, NOK, PLN, RSD, SEK & High liquidity + Europe \\
2 & AUD, ISK, NZD & Carry trade targets \\
3 & COP, TRY & Macro risk\\
4 & KRW, SGD, THB & East \& South-East Asian\\
\bottomrule
\end{tabular}
\end{table}

\begin{table}[t]
\centering
\caption{Out-of-sample realized aggregate trading costs, 2008--2023 in millions USD.}
\label{tab:realized_costs_hac}
\begin{tabular}{lrrrr}
\toprule
Architecture & $c(M)$ & $c^\text{SQ}$ & Savings & Savings (\%) \\
\midrule
HAC algorithm pool selection & 16{,}450 & 18{,}930 & 2{,}500 & 13\% \\
\bottomrule
\end{tabular}
\end{table}
 
\begin{figure}[t]
  \centering
  \begin{minipage}[t]{0.49\textwidth}
    \centering
    \includegraphics[width=0.9\linewidth]{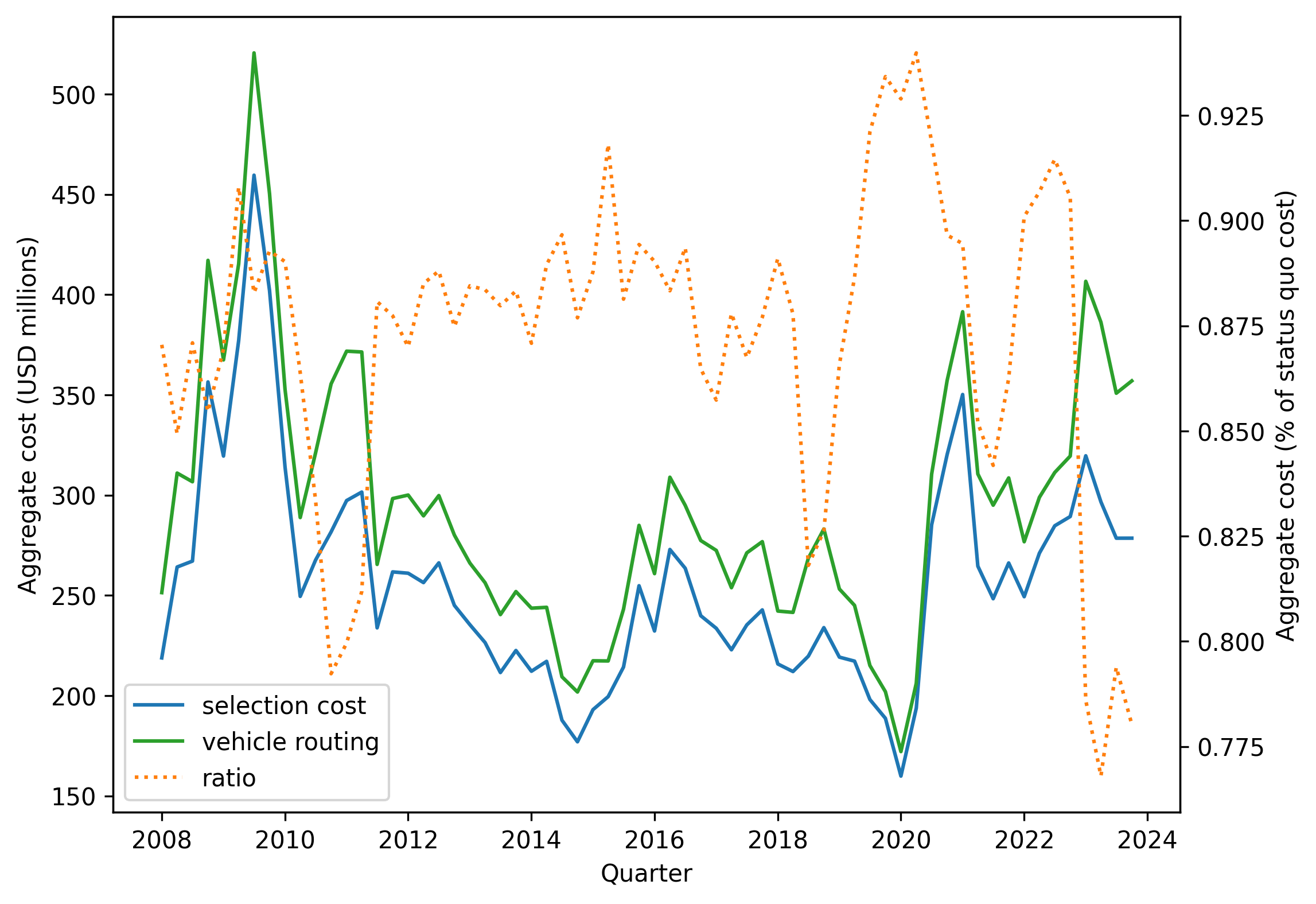}
    \caption{Quarterly realized cost of Table~\ref{tab:selected-pools_hac} architecture as a percentage of USD-routing cost.}
    \label{fig:realized_cost_hac}
  \end{minipage}%
  \hfill%
  \begin{minipage}[t]{0.49\textwidth}
    \centering
    \includegraphics[width=0.8\linewidth]{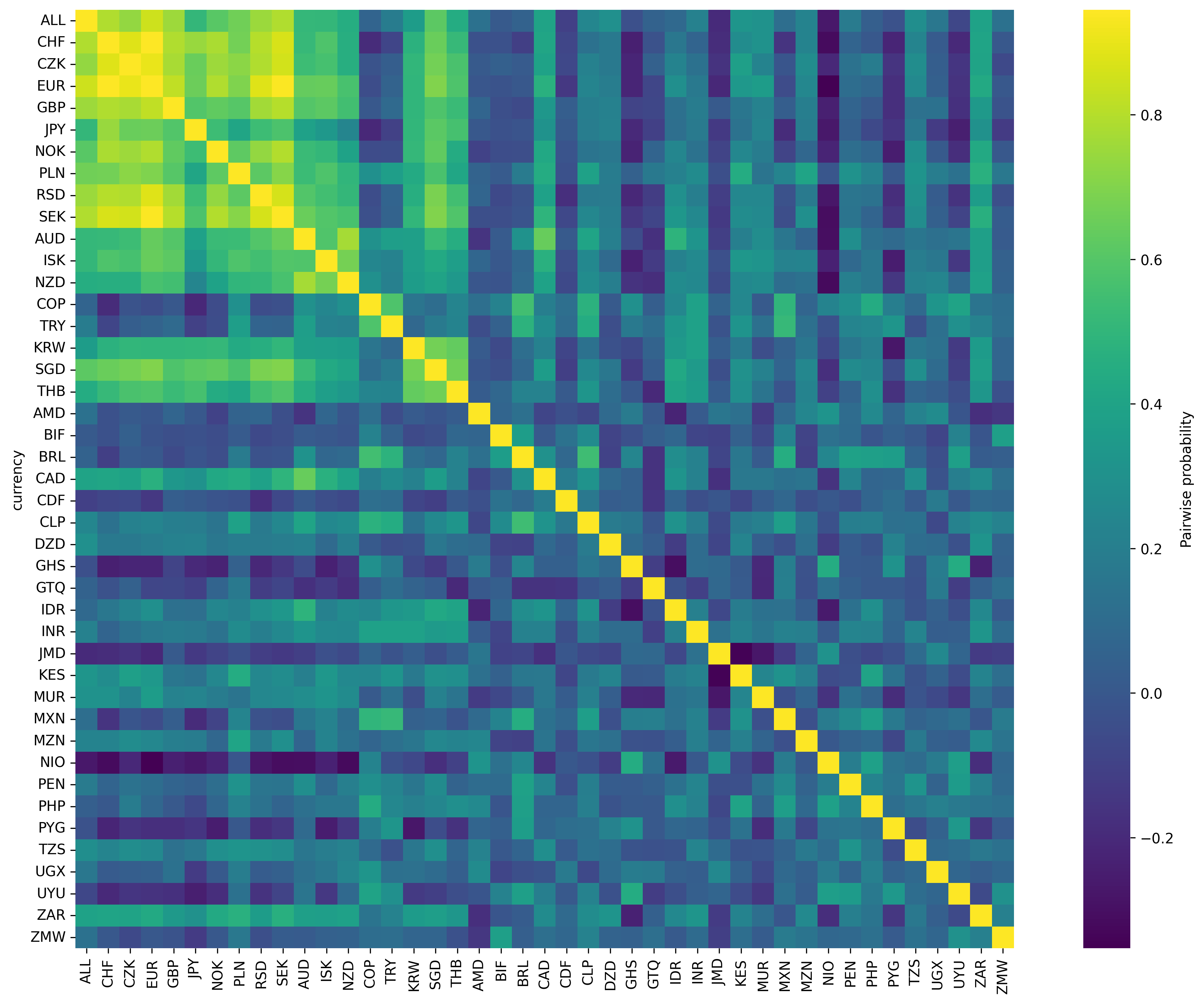}
    \caption{Heatmap of correlations with currencies ordered by HAC-optimal pool assignment.}
    \label{fig:correlations_heatmap_hac}
  \end{minipage}
\end{figure}

\section{Related Work}
\label{sec:related}

Developments on AMMs have accelerated since the constant-product invariant was introduced by Uniswap~V2~\cite{uniswapv2}. Subsequent protocols explored alternative invariants and liquidity mechanisms: Balancer generalized to the constant mean invariant, Curve developed a hybrid invariant optimized for low-volatility assets, and Uniswap~V3 introduced concentrated liquidity to improve capital efficiency~\cite{balancer,curve,uniswapv3}. Our work builds on the constant-mean framework and studies AMMs through the lens of pool formation and weight optimization.

Meanwhile, theoretical literature studies the economics of AMMs. The authors of~\cite{angeris2020improved,angeris2023replicating} show how arbitrage aligns prices with external markets and concave payoff functions can be replicated by suitable invariants, while \cite{angeris2022optimal} studies optimal routing across multi-asset AMMs. Another major focus is liquidity provider losses induced by arbitrage and price volatility, quantified by \cite{pintail2019uniswap,milionis2022automated}. Empirical and strategic analysis examine realized liquidity provider outcomes and adoption limits~\cite{heimbach2021behavior,milionis2023automated,capponi2021adoption}, optimal and equilibrium liquidity provision~\cite{cartea2024decentralised,lehar2025decentralized,aoyagi2020liquidity}, and endogenous fee setting~\cite{evans2021optimal,malinova2023learning} in 2-asset constant product pools. Closer to our setting, the work of~\cite{lipton2021automated,gogol2024cross} compare FX efficiency in AMMs under different invariants and levels of liquidity fragmentation. We add to this literature by embedding closed-form expressions for AMM prices and impermanent loss of multi-asset constant-mean pools into the equilibrium model of competitive liquidity provision from~\cite{malinova2023learning} in a foreign exchange setting.

Our setting also relates to the literature on vehicle currencies in foreign exchange markets. Consistent with our framework, \cite{krugman:1984} show that vehicle currencies emerge from transaction cost minimization and economies of scale concentrating liquidity along vehicle currency pairs, while \cite{goldberg2008vehicle} emphasize the persistence of vehicles due to trade and industry structure.

Finally, our problem is related to portfolio optimization. Building on the mean--variance optimization framework of \cite{markowitz1952portfolio}, subsequent work incorporates metaheuristic methods such as genetic and evolutionary algorithms to address nonlinearities~\cite{di2016portfolio,crama2003simulated,rubiogarcia2022portfolio}. 
Our problem, however, jointly determines pool membership and optimal weights with a highly nonlinear objective where weights also appear in the denominator.
Moreover, while portfolio optimization seeks to diversify away covariance risk, our setting benefits from highly correlated assets.

\section{Conclusion and Future Work}
\label{sec:conclusion}
This paper studies the design of multi-currency AMM architectures for
foreign exchange systems. In our framework, the equilibrium
trading cost depends jointly on trading volumes, exchange-rate volatilities, and pool composition. We show that under this framework,
pooling multiple currencies in a constant weighted mean AMM can reduce aggregate trading costs compared to the
status quo---bilateral routing through a vehicle currency.
The gains arise because multi-currency pools enable direct trades while still concentrating liquidity across many trading pairs.
While multi-asset pools suffer from increased risk exposure, pooling assets with correlated exchange rates and appropriately weighing exposures limits the additional impermanent loss.
We derive analytical results characterizing optimal weights in
symmetric environments and a tractable approximation for general weight optimization, and formulate the system-level problem of partitioning currencies into pools.

To solve the partitioning problem,
we employ a hierarchical agglomerative clustering heuristic algorithm
that exploits the correlation structure of exchange-rate returns to efficiently group currencies into structurally similar pools. Empirically, using exchange-rate and international trade
data for 43 non-pegged currencies over 2008--2023, we find that the resulting architecture
reduces realized aggregate trading costs by \textasciitilde 13\% relative to bilateral USD
routing, with improvements remaining stable through major episodes of financial stress.

Several directions remain for future research. First, relaxing the assumptions we make by allowing persistent order imbalances, richer exchange-rate dynamics, stochastic trade-size distributions, or endogenous trading demand may materially affect the optimal architecture and equilibrium costs. Second, future work could seek analytical characterizations of the general weight optimization solution or derive approximation schemes with formal error bounds.
Third, the present framework treats pool membership as fixed after the initial pool setup. The design of a dynamic multi-currency AMM that supports rebalancing mechanisms allowing currencies to migrate between pools as market conditions change could improve the flexibility and robustness of the multi-currency architecture, especially for currencies issued by countries with volatile economic policies. 
The current work uses historical correlations for pool formation, but the same theory extends to dynamic pools formed using time-varying volatilities.
Finally, the robustness of the agglomerative clustering algorithm should be studied over a variety of training data regimes and sets of currencies before it can be proposed as a reliable method for currency pooling, especially if deployed in an algorithmically rebalancing multi-currency AMM.
Indeed, there may exist algorithms that are
more efficient or robust across different market environments. Alternative approaches based on graph partitioning, metaheuristics, convex relaxations, or machine learning may yield improved solutions or stronger theoretical guarantees. These methods may also be able to deal with more complex pooling architectures, such as non-disjoint pools or pools without a common vehicle currency.

\cleardoublepage 
\phantomsection
\addcontentsline{toc}{section}{References}
\bibliography{references}

\appendix
\section{Proofs}
\label{app:proofs}

\subsection{Proofs of benchmark arrangement trading costs}
\label{app:benchmarks}
\statusquocost*
\begin{proof}
    We sum the costs over the volume of the bilateral trades:
\[
\sum_{i<j} Q_{ij}\left(\sqrt{\frac{\Delta \sigma_i^2}{E[Q_i]}}
+ \sqrt{\frac{\Delta \sigma_j^2}{E[Q_j]}}\right) = \sum_{i<j} Q_{ij} \sqrt{\frac{\Delta \sigma_i^2}{E[Q_i]}}
+ \sum_{i<j} Q_{ij} \sqrt{\frac{\Delta \sigma_j^2}{E[Q_j]}}.
\]
Swapping $i \leftrightarrow j$ in the second sum and using $Q_{ij} = Q_{ji}$,
the two combine to give $\sum_i \sum_{j \neq i} \frac{\sigma_i Q_{ij} \sqrt{\Delta}}{\sqrt{E[Q_i]}}$.
Hence, in expectation,
\[
c^\text{SQ} = \sum_i \sqrt{\sigma_i^2 E[Q_i] \Delta} \hfill \qedhere
\]
\end{proof}
\equalweightscost*
\begin{proof}
    With equal weights, we can drop the weights and recover a constant product invariant. For a trade of size $\Delta$ in numeraire from currency $i$ to $j$, this invariant function requires
\[
  \left(R_i + \frac{\Delta}{S_i}\right)
  \left(R_j - \frac{\Delta}{S^{\text{eff}}}\right)
  = R_i R_j,
\]
\[
\iff \frac{\Delta}{S^\text{eff}} = R_j - \frac{R_i R_j}{R_i + \Delta/S_i} = R_j\left(\frac{\Delta/S_i}{R_i + \Delta/S_i}\right) = \frac{R_j \Delta}{S_i R_i + \Delta} \iff
S^\text{eff} = \frac{S_i R_i + \Delta}{R_j}
\]
The spot price of currency $i$ in terms of currency $j$ implied by the pool is
$S_{ij}^{\text{AMM}} = R_j / R_i$.  At fundamental prices,
$S_{ij}^{\text{AMM}} = S_i / S_j$ implies $S_i R_i = S_j R_j$ for all pairs, so
each currency contributes equal value to the pool: $V = (N+1) S_i R_i$ for all $i$.

Hence, $S^\text{eff} = \frac{S_i R_i}{R_j} + \frac{\Delta}{R_j} = S_j + \frac{\Delta (N+1) S_j}{V}$. The price impact $\Pi(\Delta) = \frac{(N+1)\Delta}{V}$ follows from Equation~\eqref{eq:priceimpact} using $S = S_j$.

The pool value $V(T)$ is also governed by the invariant: $\prod_i R_i(0) = \prod_i R_i(T)$, implying
\[
\prod_i \frac{V(0)}{S_i(0)(N+1)} = \prod_i \frac{V(T)}{S_i(T)(N+1)} \iff \frac{V(T)}{V(0)} = \prod_i \left(\frac{S_i(T)}{S_i(0)}\right)^{1/(N+1)} =\prod_i x_i^{1/(N+1)}.
\]
Meanwhile, the counterfactual hold value is 
\[
V_\text{hold}(T) = \sum_i x_i S_i(0) R_i(0) = \frac{V(0)}{N+1} \sum_i x_i
\]
By Equation~\eqref{eq:ipl}, the $\IPL$ is then
\[
  \IPL = \prod_i x_i^{1/(N+1)} - \frac{1}{N+1}\sum_i x_i 
\] 
For $x_i = 1 + \epsilon_i$ and $\epsilon_i \approx 0$, the second-order Taylor approximation around $\epsilon_i = 0$ is
\[
\IPL \approx -\frac{1}{2(N+1)^2} \sum_{i<j}(\epsilon_i - \epsilon_j)^2, \qquad E[|\IPL |] \approx \frac{H}{2(N+1)^2}, \qquad H = \sum_{i<j} \sigma_{ij}^2
\]
where the expectation is taken assuming $\epsilon_i$ has zero-mean.
This results in the break-even pool value $V = \frac{2(N+1)^2 f E[Q]}{H}$ (from Equation~\eqref{eq:breakeven}) and unit trading cost
\[
c(\Delta,f) = \frac{H \Delta}{2(N+1) f E[Q]} + f
\]
Minimizing over $f$ (the AM-GM inequality implies that $A/f + Bf \geq 2\sqrt{AB}$ for $f > 0$, and equality here holds when $f = \sqrt{A/B}$) 
yields
\[
f^* = \sqrt{\frac{H\Delta}{2(N+1) E[Q]}}, \qquad c^* = \sqrt{\frac{2 H \Delta}{(N+1) E[Q]}}
\] 
Aggregating over the total expected trade volume $E[Q]$ yields the expected aggregate cost
\[
c^\text{SYM} = \sqrt{\frac{2 H E[Q] \Delta}{(N+1)}} \hfill \qedhere
\]
\end{proof}

\subsection{Proofs of costs in constant-mean AMMs}
\label{app:ammcosts}
\priceimpact*
\begin{proof}
    Consider a trade between currencies $i$ and $j$ of notional size $\Delta$ in
the numeraire (so the trader pays $\Delta / S_i$ units of currency $i$).  
As only the reserves $R_i$ and $R_j$ change, the invariance function requires
\[
  \left(R_i + \frac{\Delta}{S_i}\right)^{w_i}
  \left(R_j - \frac{\Delta}{S^{\text{eff}}}\right)^{w_j}
  = R_i^{w_i} R_j^{w_j},
\]
where $\Delta / S^{\text{eff}}$ is the number of units of $j$ received.
Isolating the output quantity:
\[
  \frac{\Delta}{S^{\text{eff}}}
  = R_j \!\left(1 - \left(1 + \frac{\Delta}{S_i R_i}\right)^{-w_i/w_j}\right)
\]
Assuming $\Delta \ll S_i R_i = w_i V$, a second-order Taylor expansion gives
\[
  \left(1 + \frac{\Delta}{S_i R_i}\right)^{-w_i/w_j}
  \approx 1 - \frac{w_i}{w_j} \frac{\Delta}{S_i R_i}
    \left(1 - \frac{1}{2} \frac{\Delta}{S_i R_i} \frac{w_i + w_j}{w_j}\right)
\]
The spot price of currency $i$ in terms of currency $j$ is $S_{ij}^{\text{AMM}} = (R_j / w_j) / (R_i / w_i)$. At fundamental prices, $S_{ij}^{\text{AMM}} = S_i / S_j$ implies $S_i R_i / w_i = S_j R_j / w_j$ for all pairs, so each currency $i$ accounts for $w_i$ of total pool value, i.e., $w_i V = S_i R_i$. Then, using $R_j w_i / w_j = (S_i / S_j) R_i$ and $S_i R_i = w_i V$:
\[
  \frac{\Delta}{S^{\text{eff}}}
  \approx \frac{\Delta}{S_j}
    \left(1 - \frac{\Delta}{2V} \frac{w_i + w_j}{w_i w_j}\right)
\]
Inverting and applying a first-order expansion (valid for
$\Delta \ll 2 w_i w_j V / (w_i + w_j)$):
\[
  S^{\text{eff}} \approx S_j
    \left(1 + \frac{w_i + w_j}{w_i w_j} \frac{\Delta}{2V}\right)
\]
Substituting into~\eqref{eq:priceimpact} yields the result. \qedhere
\end{proof}
\impermanentloss*
\begin{proof}
    We first derive the exact expression for $V(T)$. As derived in the proof of Proposition~\ref{prop:priceimpact}, at fundamental prices, $R_i(t) = w_i V(t) / S_i(t)$ for all $i$.  Substituting into the constant-mean
invariant:
\[
  \prod_i \!\left(\frac{w_i V(T)}{S_i(T)}\right)^{w_i}
  = \prod_i \!\left(\frac{w_i V(0)}{S_i(0)}\right)^{w_i}
\]
Since $\sum_i w_i = 1$, the $V$ terms factor out as $V(T)^{\sum_i w_i} = V(T)$
and $V(0)^{\sum_i w_i} = V(0)$, giving
\[
  V(T) = V(0) \prod_i \left(\frac{S_i(T)}{S_i(0)}\right)^{w_i}
       = V(0) \prod_i x_i^{w_i}
\]
The counterfactual hold value is
\[
  V_{\text{hold}}(T)
  = \sum_i S_i(T) R_i(0)
  = \sum_i S_i(T) \cdot \frac{w_i V(0)}{S_i(0)}
  = V(0) \sum_i w_i x_i
\]
Substituting into~\eqref{eq:ipl} gives the exact expression
$\IPL = \prod_i x_i^{w_i} - \sum_i w_i x_i$.
 
For $x_i = 1+\epsilon_i$ and $\epsilon_i \approx 0$, the second-order Taylor
expansion around $\epsilon_i = 0$ is
\begin{align*}
  \IPL & \approx
  \left(1 + \sum_i w_i \epsilon_i
  + \frac{1}{2}\!\left[\left(\sum_i w_i \epsilon_i\right)^2
  - \sum_i w_i \epsilon_i^2\right]\right)
  - \left(1 + \sum_i w_i \epsilon_i\right)\\
  &= -\frac{1}{2}\!\left(\sum_i w_i \epsilon_i^2
  - \left(\sum_i w_i \epsilon_i\right)^2\right)
\end{align*}
Expanding $(\sum_i w_i \epsilon_i)^2 = \sum_i w_i^2 \epsilon_i^2
+ 2 \sum_{i<j} w_i w_j \epsilon_i \epsilon_j$, grouping the sums of $\epsilon_i^2$, and using $1 - w_i =
\sum_{j \neq i} w_j$ to write $\sum_i w_i(1-w_i)\epsilon_i^2
= \sum_{i<j} w_i w_j(\epsilon_i^2 + \epsilon_j^2)$, we get
\[
  \IPL \approx
  -\frac{1}{2} \sum_{i<j} w_i w_j
  (\epsilon_i^2 + \epsilon_j^2 - 2\epsilon_i \epsilon_j)
  = -\frac{1}{2} \sum_{i<j} w_i w_j (\epsilon_i - \epsilon_j)^2
\]
The expectation under zero-mean returns is
$E[|\IPL|] \approx H_w / 2$ where $H_w = \sum_{i<j} w_i w_j \sigma_{ij}^2$.%
\qedhere
\end{proof}
\cost*
\begin{proof}
    By Lemma~\ref{prop:ipl}, $E[|\IPL|] \approx H_w / 2$.  The
break-even condition~\eqref{eq:breakeven} therefore gives equilibrium pool depth
\[
  V = \frac{2 f E[Q]}{H_w}.
\]
By Lemma~\ref{prop:priceimpact}, the unit trading cost for a pair $(i,j)$
trade of size $\Delta$ is
\[
  c_{ij}(\Delta, f)
  = \frac{w_i + w_j}{w_i w_j} \frac{\Delta}{2V} + f
  = \frac{w_i + w_j}{w_i w_j} \frac{H_w \Delta}{4 f E[Q]} + f.
\]
Aggregating over all pairs, weighted by their respective volumes $Q_{ij}$:
\[
  c(\Delta, f)
  = \frac{H_w \Delta}{4 f E[Q]}
    \sum_{i < j} Q_{ij} \frac{w_i + w_j}{w_i w_j} + f Q.
\]
We simplify $\sum_{i<j} Q_{ij}(w_i + w_j)/(w_i w_j)
= \sum_{i<j} Q_{ij}/w_j + \sum_{i<j} Q_{ij}/w_i$.
Swapping $i \leftrightarrow j$ in the second sum and using $Q_{ij} = Q_{ji}$,
both sums combine to give $\sum_i Q_i / w_i$, where
$Q_i = \sum_{j \neq i} Q_{ij}$ is the total volume involving currency $i$.
Hence
\begin{equation}\label{eq:exantecost}
     c(\Delta, f) = \frac{H_w \Delta}{4 f E[Q]} \sum_i \frac{Q_i}{w_i} + f Q.
\end{equation}
In expectation, minimizing over $f$ (the AM-GM inequality implies that 
$A/f + Bf \geq 2\sqrt{AB}$ for $f > 0$, and equality here holds when $f = \sqrt{A/B}$) 
yields
\[
f^* = \sqrt{\frac{H_w \Delta}{4 (E[Q])^2}
      \sum_i \frac{E[Q_i]}{w_i}},
    \qquad
    c^* = \sqrt{\Delta H_w \sum_i \frac{E[Q_i]}{w_i}} \hfill \qedhere
\]
\end{proof}

\subsection{Proofs of parametrized optimums.}
\label{app:param}
\sym*
\begin{proof}
    With $w_1 = w_2 = (1-w_0)/2$:
\[
  H_w
  = w_1 w_2 \sigma_{12}^2 + w_0 w_1 \sigma_{01}^2 + w_0 w_2 \sigma_{02}^2
  = \tfrac{1}{4}(1-w_0)^2 s^2 \sigma^2 + (1-w_0) w_0 \sigma^2.
\]
The volume-weighted sum is
$\sum_i E[Q_i]/w_i = 4q(1+v)/(1-w_0) + 2q/w_0$.
After substituting into Equation~\eqref{eq:aggregatecost} and simplifying, the expected aggregate cost $c(w_0;s,v)$ satisfies
\[
  \frac{c(w_0;s,v)^2}{2 q \sigma^2 \Delta}
  = \left(\frac{s^2 v}{2} + 1\right) + P(w_0; s, v),
\]
where
\[
  P(w_0; s, v)
  = w_0 \!\left(1 - \tfrac{s^2}{4}\right)(1+2v)
    + \frac{s^2}{4 w_0}.
\]
Here, $c$ is monotonically increasing in $P$, so we solve the equivalent problem of minimizing $P$.
Since $d^2P/dw_0^2 = s^2 / (2 w_0^3) > 0$ for all $w_0 > 0$, $P$ is strictly convex in $w_0$ for that region.
We have the form $A/w_0 + Bw_0$, which has minimum $2\sqrt{AB}$ attained at $w_0=\sqrt{A/B}$,
yielding
\[
w_0^* = \frac{s}{\sqrt{(1+2v)(4-s^2)}}, \qquad P(w_0^*;s,v) = \frac{s}{2}\sqrt{(4-s^2)(1+2v)},
\]
When $s^2 < 2(1+2v)/(1+v)$, $w_0^* < 1$ is a valid interior solution with
\[
  c^* = \sqrt{q \sigma^2 \Delta}
    \sqrt{s^2 v + 2 + s\sqrt{(4-s^2)(1+2v)}}.
\]
For $s \in [0,2]$ and $v>0$, these quantities are well-defined. 

When $s^2 \geq 2(1+2v)/(1+v)$, the unconstrained minimizer is $w_0^* \geq 1$, and the convexity of $P(w_0)$ implies that the constrained infimum is attained as $w_0 \rightarrow 1$. \qedhere
\end{proof}
\nsym*
To prove the proposition, we first establish the following lemma on the admissible range of $s$.
{\setlength{\topsep}{0pt}
\begin{lemma}\label{lemma:parameter_range}
    Under the above parametrization, $s \in \left[0,\sqrt{\frac{2N}{N-1}}\right]$.
\end{lemma} 
}
\begin{proof}
    Let $C$ be the $N \times N$ covariance matrix of the returns of the $N$ symmetric currencies, which has $\sigma^2$ on the diagonal and $\sigma^2 - \frac{s^2\sigma^2}{2}$ elsewhere. The variance of any linear combination of the currencies must be non-negative, so $R$ must be positive semi-definite (PSD). For simplicity, we will work with the correlation matrix $R$ which satisfies $\sigma^2 R = C$ and must also be PSD.
    
    We can write $R = (1-\rho)I + \rho uu^T$, where $I$ is the $N \times N$ identity matrix and $u^T = (1,\ldots,1)$. Then, the eigenvectors and eigenvalues of $R$ satisfy
    \[
    Rv = (1-\rho)v + \rho \langle u,v\rangle u= \lambda v \iff (\lambda + \rho - 1)v = \rho \langle u,v \rangle u.
    \]
    which implies that, either both sides are $0$, which corresponds to eigenvalue $\lambda_1 = 1-\rho$, or $v = au$ for some $a \in \mathbb{R}\setminus \{0\}$ and $\lambda + \rho - 1 = \rho N$,
    for which we have eigenvalue $\lambda_2 = \rho(N-1)+1$. Thus $R$ is positive semidefinite if and only if
    \[
    1 - \rho \ge 0 \quad \text{and} \quad 1 + (N-1)\rho \ge 0 \iff -\frac{1}{N-1} \leq 1-\frac{s^2}{2} \leq 1.
    \]
    Rearranging, we have $s \leq \sqrt{\frac{2N}{N-1}}$. \qedhere
\end{proof}
Now, assuming $s \in [0,\sqrt{\frac{2N}{N-1}}]$, we prove the proposition.
\begin{proof}
    Now, in the $N$ symmetric currency AMM, we have $w_i = (1-w_0)/N$, implying
\[
  H_w
  = \sum_{i > 0} \frac{w_0(1\!-\!w_0)}{N} \sigma^2 + \sum_{0<i<j} \frac{(1\!-\!w_0)^2}{N^2} s^2 \sigma^2 
  = N\frac{w_0(1\!-\!w_0)}{N} \sigma^2 + \binom{N}{2} \frac{(1\!-\!w_0)^2}{N^2} s^2 \sigma^2.
\]
This simplifies to
\[
  H_w = w_0(1\!-\!w_0)\sigma^2 + \frac{N\!-\!1}{2N} (1\!-\!w_0)^2 s^2 \sigma^2
\]
The volume-weighted sum is
\[
\sum_i \frac{E[Q_i]}{w_i} = \frac{Nq}{w_0} + \frac{N(q + vq(N\!-\!1))}{(1\!-\!w_0)/N} = Nq\left(\frac{(1\!-\!w_0) + w_0 \cdot N(1+v(N-1))}{w_0(1\!-\!w_0)}\right)
\]
Substituting, the expected aggregate cost is
\begin{equation}\label{eq:paramcostw0}
    c(\Delta,f)
  = \sqrt{q \sigma^2 N \Delta} \sqrt{\left(w_0 + \frac{N-1}{2N} (1-w_0) s^2\right)\left(\frac{(1-w_0) + w_0 \cdot N(1+v(N-1))}{w_0}\right)}
\end{equation}
and, in optimum, satisfies
{%
\setlength{\belowdisplayskip}{8pt}
\setlength{\belowdisplayshortskip}{8pt}
\begin{flalign*}
  \frac{(c^*)^2}{q \sigma^2 N \Delta} &= (1\!-\!w_0)+ w_0 N(1\!+\!v(N\!-\!1)) + \frac{N\!-\!1}{2Nw_0}(1\!-\!w_0)^2s^2 + \frac{N\!-\!1}{2} (1\!-\!w_0)s^2 (1\!+\!v(N\!-\!1)) && \\
  & = 1+\frac{s^2}{2N}(N-1)(N(1+v(N-1))-2) + P(w_0;s,v) &&
\end{flalign*}
}%
{%
\setlength{\belowdisplayskip}{5pt}%
\setlength{\belowdisplayshortskip}{5pt}%
where
\[
    P(w_0; s, v) = w_0(N-1)(1+Nv)\frac{1}{2N}(2N-(N-1)s^2) + \frac{(N-1)s^2}{2Nw_0}
\]
}%
Again, $c$ is strictly monotonically increasing in $P$, we minimize $P$ instead.
Since $d^2P/dw_0^2 = (N-1)s^2 / w_0^3 > 0$ for all $w_0 > 0$, $P$ is strictly convex in $w_0$ for that region.
Once again, we have the form $A/w_0 + Bw_0$, which attains its minimum, $2\sqrt{AB}$ at $w_0 = \sqrt{A/B}$, yielding
\[
w_0^* = \frac{s}{\sqrt{(1+Nv)(2N-(N-1)s^2)}}, \;\;\; P(w_0^*; s, v) = \frac{s(N-1)}{N}\sqrt{(1+Nv)(2N-(N-1)s^2)},
\]
When $s^2 < 2(1+2v)/(1+v)$, $w_0^* < 1$ is a valid interior solution with
\begin{flalign}\label{eq:paramcost}
  c^*\!=\!\sqrt{q \sigma^2 \Delta}\sqrt{N\!+\!\frac{s^2}{2}(N\!-\!1)(N(1\!+\!v(N\!-\!1))\!-\!2)+\!s(N\!-\!1)\!\sqrt{(1\!+\!Nv)(2N\!-\!(N\!-\!1)s^2)}}  &&
\end{flalign}
These quantities are well-defined for $v>0$ and the range of $s$ derived in Lemma~\ref{lemma:parameter_range}. 

When $s^2 \geq  2(1+Nv)/(1+(N-1)v)$, the unconstrained minimizer is $w_0^* \geq 1$, and the convexity of $P(w_0)$ implies that the constrained infimum is attained as $w_0 \rightarrow 1$. \qedhere
\end{proof}

\subsection{Proof of tractable weight optimization}
\label{app:approx}
\approxi*
\begin{proof}
    We minimize $\tilde{\mathcal{F}}(w) = \sum_i \frac{E[Q_i]}{w_i H_i}$ subject to
$\sum_i w_i = 1$, $w_i > 0$, where $H_i = \sum_{j\neq i} \sigma_{ij}^2$, via the Lagrangian
\[
  \mathcal{L}(w, \lambda)
  = \sum_i \frac{E[Q_i]}{w_i H_i} + \lambda \!\left(\sum_i w_i - 1\right).
\]
The first-order conditions $\nabla_w \mathcal{L}(w,\lambda) = 0$ gives
$\frac{E[Q_i]}{w_i^2 H_i} = \lambda$, so $w_i = \sqrt{\frac{E[Q_i]}{\lambda H_i}}$ for all $i$.
Imposing $\sum_i w_i = 1$ through the condition $\nabla_\lambda \mathcal{L}(w,\lambda) = 0$:
\[
  \frac{1}{\sqrt{\lambda}} \sum_j \sqrt{E[Q_j]/H_j} = 1
  \implies \sqrt{\lambda} = \sum_j \sqrt{E[Q_j]/H_j} \equiv D, \implies w_i^* = D^{-1}\sqrt{E[Q_i]/H_i}
\]
The hessian $\nabla_{ww}^2 \mathcal{L}(w,\lambda)$ has 
$\frac{\partial^2 \mathcal{L}}{\partial w_i^2} = \frac{2E[Q_i]}{w_i^3 H_i}$ on the diagonal
and $\frac{\partial^2 \mathcal{L}}{\partial w_i \partial w_j} = 0$ elsewhere, so the second-order 
sufficient condition, $v^T \nabla_{w w}^2 \mathcal{L}(w,\lambda) v > 0$ 
for all $v \neq 0$, becomes
\[
\sum_i v_i^2 \frac{\partial^2 \mathcal{L}}{\partial w_i^2} = \sum_i v_i^2 \frac{2E[Q_i]}{w_i^3 H_i} > 0,
\]
which holds when all $w_i>0$, so this a global minimum for $w_i>0$. \qedhere
\end{proof}

\iffalse

{
\renewcommand{\thesection}{B--D}
\section{Full Version Material}

\makeatletter
\protected@edef\@currentlabel{B}
\phantomsection
\label{app:thresholds}

\protected@edef\@currentlabel{C}
\phantomsection
\label{app:systemcost}

\protected@edef\@currentlabel{D}
\phantomsection
\label{app:greedy}
\makeatother
}

Appendices~\ref{app:thresholds},~\ref{app:systemcost}, and~\ref{app:greedy} are omitted here due to space limitations. They can be found in their respective Appendix sections in the full version of this paper~\cite{full}.

\else

\section{Thresholds for Parametrized Case}
\label{app:thresholds}

\begin{proposition}
    Under the symmetric $(v,s)$ parametrization for three currencies,
\begin{itemize}
    \item $c^* \leq c^\text{SQ}$ whenever $s < \sqrt{\frac{2(1+2v)}{1+v}}$, a threshold equal to $\sqrt{2}$ at $v=0$ and increasing to $2$ as $v\to\infty$.
    \item $c^* < c^\text{BP}$ for all $v>0, s\in\left[0,\sqrt{\frac{2(1+2v)}{1+v}}\right]$. Additionally, $c^\text{BP} < c^\text{SQ}$ whenever $s < \frac{2(\sqrt{1+v}-1)}{\sqrt{v}}$, a threshold equal to $0$ at $v\to 0^+$ and increasing to $2$ as $v\to\infty$.
    \item $c^* \leq c^\text{SYM}$ for all $v>0, s\in[0,2]$, with equality only when $s=v=1$. Additionally, $c^\text{SYM} < c^\text{SQ}$ whenever $s < \sqrt{\frac{2(1+2v)}{2+v}}$, a threshold equal to $1$ at $v=0$ and increasing to $2$ as $v\to\infty$.
\end{itemize}
\end{proposition} 

\begin{proof}
    Plugging $\sigma, s, q, v$ into $c^\text{SQ}$, $c^\text{BP}$, and $c^\text{SYM}$ gives
\[
c^\text{SQ}(s,v) = \sqrt{q\sigma^2 \Delta }\sqrt{4(1+v)}, \qquad
c^\text{BP}(s,v) = \sqrt{ q \sigma^2 \Delta} (2+s\sqrt{v})
\]
\[
c^\text{SYM}(s,v) = \sqrt{q \sigma^2 \Delta} \sqrt{\tfrac{2}{3}(2+s^2)(2+v)}.
\]
\begin{itemize}
    \item For a given $s,v$, when $w_0 \rightarrow 1$, the expected aggregate cost approaches 
    \[
    c(w_0;s,v) \rightarrow \sqrt{2 q \sigma^2 \Delta}\sqrt{\frac{s^2 v}{2}+1 + P(1;s,v)} = \sqrt{q \sigma^2 \Delta} \sqrt{4(1+v)} = c^\text{SQ}(s,v).
    \]
    This means that $c^* \leq c^\text{SQ}$ if and only if $c(w_0;s,v)$ has a minimum in $(0,1)$. That is, whenever $0 <w_0^* = \frac{s}{\sqrt{(1+2v)(4-s^2)}} < 1$, or $s^2 < \frac{2(1 + 2v)}{1+v}$.\\
    \item $c^*$ is well-defined when $s < \sqrt{\frac{2(1+2v}{1+v}}$ and
    $c^* < c^\text{BP}$ whenever 
    \[
    s^2v + 2 + s\sqrt{(4-s^2)(1+2v)} < (2+s\sqrt{v})^2 = 4 + 4s\sqrt{v} + s^2 v
    \]
    Rearranging, we get
    \[
    s\sqrt{(4-s^2)(1+2v)} < 2 + 4s\sqrt{v}
    \]
    Squaring both sides and expanding,
    \[
    4s^2+ 8vs^2 - s^4 - 2s^4v < 4 + 16s\sqrt{v} + 16s^2 v
    \]
    Bringing everything to the right-hand side,
    \[
    16s\sqrt{v} + 8s^2 v + 2s^4 v + s^4 - 4s^2 + 4 = 16s\sqrt{v} + 8s^2 v + 2s^4 v + (s^2-2)^2 > 0
    \]
    Each term is positive for all $v >0$ and $s\in\left[0,\sqrt{\frac{2(1+2v)}{1+v}}\right]$, so the inequality holds for all valid $v,s$. 
    The threshold for $c^\text{BP} < c^\text{SQ}$ follows directly from the definitions of $c^\text{BP}$ and $c^\text{SQ}$ by isolating $s$.\\
    \item The inequality holds because $c^*$ is the minimum over all valid weights. The threshold for $c^\text{SYM} < c^\text{SQ}$ follows directly from the definitions of $c^\text{SYM}$ and $c^\text{SQ}$ by isolating $s$. \qedhere
\end{itemize}
\end{proof}
\begin{proposition}
Under the symmetric $(v,s)$ parametrization for $N+1$ currencies,
\begin{itemize}
  \item $c^* \leq c^\text{SQ}$ whenever $s < \sqrt{\frac{2(1 + Nv)}{1+(N-1)v}}$. This threshold is $\sqrt{2}$ when $v = 0$ and increases to $\sqrt{\frac{2N}{N-1}}$ as $v \rightarrow \infty$.
  \item  $c^\text{BP} < c^\text{SQ}$ whenever $s < \frac{2(\sqrt{1+(N-1)v}-1)}{(N-1)\sqrt{v}}$, which is $0$ at $v \rightarrow 0^+$ and approaches $\frac{2}{\sqrt{N-1}}$ as $v \rightarrow \infty$.
  \item $c^* \leq c^\text{SYM}$ follows from the weight optimization problem, while $c^\text{SYM} < c^\text{SQ}$ whenever $s \leq \sqrt{\frac{2(1+Nv)}{2+(N-1)v}}$, which is $1$ when $v=0$ and approaches $\sqrt{\frac{2N}{N-1}}$ as $v \rightarrow \infty$.
\end{itemize}    
\end{proposition}

\begin{proof}
Plugging $N, \sigma, s, q, v$ into $c^\text{SQ}$, $c^\text{BP}$, and $c^\text{SYM}$ gives
\[
c^\text{SQ}(s,v) = \sqrt{q\sigma^2 \Delta }\sqrt{N^2(1+v(N-1))}, \quad
c^\text{BP}(s,v) = \sqrt{ q \sigma^2 \Delta} \!\left(N+s\sqrt{v}\tbinom{N}{2}\right)
\]
\[
c^\text{SYM}(s,v) = \sqrt{q \sigma^2 \Delta} \sqrt{\tfrac{2}{N+1}\!\left(N+s^2\tbinom{N}{2}\right)\!\left(N+v\tbinom{N}{2}\right)}.
\]
\begin{itemize}
    \item For a given $s,v$, when $w_0 \rightarrow 1$, the expected aggregate cost in Equation~\eqref{eq:paramcostw0}
    
    \hspace*{-\leftmargin}
    \begin{minipage}{\dimexpr\linewidth\relax}
    \begin{flalign*}
        c(w_0;s,v)\!\rightarrow\!\sqrt{N q \sigma^2 \Delta} \sqrt{(1\!+\!0)\!\left(\frac{0\!+\!1\!\cdot\! N(1\!+v(N\!-\!1))}{1}\right)} & = \sqrt{q \sigma^2 \Delta} \sqrt{N^2(1\!+\!v(N\!-\!1))} \\
        &= c^\text{SQ}(s,v)
    \end{flalign*}
    \end{minipage}
    
    This means that $c^* \leq c^\text{SQ}$ if and only if $c(w_0;s,v)$ has a minimum in $(0,1)$. That is, whenever $0 <w_0^* = \frac{s}{\sqrt{(1+Nv)(2N-(N-1)s^2)}} < 1$, or $s^2 < \frac{2(1 + Nv)}{1+(N-1)v}$.

    \item The thresholds for $c^\text{BP} < c^\text{SQ}$ and $c^\text{SYM} < c^\text{SQ}$ follow directly from the definitions of $c^\text{BP}$, $c^\text{SQ}$, and $c^\text{SYM}$ by isolating $s$. \qedhere
\end{itemize}
\end{proof}

Each of the thresholds derived above approach a limit as $v \rightarrow \infty$. These limits vary with $N$ and are plotted in Figure~\ref{fig:thresholds}.

\begin{figure}
    \centering
    \includegraphics[width=0.5\linewidth]{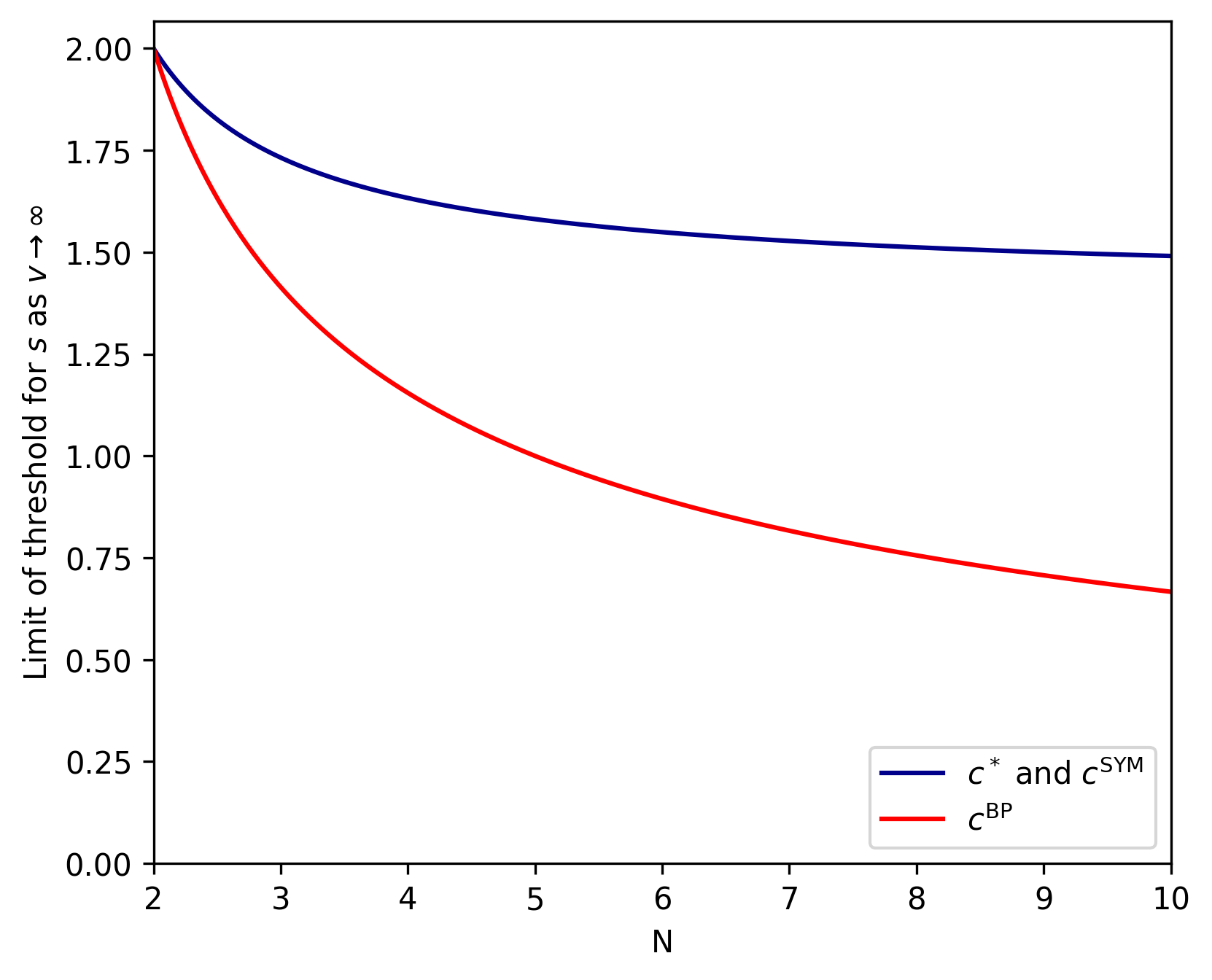}
    \caption{Threshold limits as $v \rightarrow \infty$ for different values of $N$.}
    \label{fig:thresholds}
\end{figure}

\section{System Cost Under Currency Partition}\label{app:systemcost}

Consider a pool $I = G \cup \{0\}$.
Since $i \in G$ trades all of its volume through the pool $I$ (either directly to another pool member, or to an external currency routed through currency $0$), the effective pool volumes are $Q^I_i = Q_i$ for all $i \in G$ and $Q_0^I = \sum_{i \in G} (Q_{0i} + \sum_{j \not\in I} Q_{ij})$, the sum of the trades between $0$ and pool members and the trades from a pool member to an external currency, routed through $0$.
The cost (Proposition~\ref{prop:cost}) for pool $I$ is, under ex-ante expectations,
\begin{equation*}
  c(I) = \sqrt{H_w^I \Delta \sum_{i \in I} \frac{E[Q^I_i]}{w_i}},
  \qquad
  H_w^I = \sum_{\substack{i < j,\; i,j \in I}} w_i w_j \sigma_{ij}^2,
\end{equation*}
Each unpooled currency $i \in O$ routes all its trade volume through a bilateral pool with currency $0$, incurring an aggregate cost (Proposition~\ref{prop:sqcost}) of $c^\text{SQ}(O) = \sum_{i \in O} \sqrt{\sigma_i^2 E[Q_i] \Delta}$. The total expected and realized (for quarter $t$) system costs are then, respectively,
\begin{equation}\label{eq:systemcost}
  c(M) = c^\text{SQ}(O) + \sum_{I \in M} c(I), \qquad c(M,t) = c^\text{SQ}(O,t) + \sum_{I \in M} c(I,t)
\end{equation}
where, for the minimizing fee $f_I^*$ that results in cost $c(I)$, weights $w_I^* = \arg\min_{w_I} c(I)$, and liquidity $V_I^*$ satisfying the break-even condition~\eqref{eq:breakeven},
\[
  c(I,t)=\frac{\Delta}{2V_I^*}\sum_{i\in I}\frac{Q^I_i(t)}{{w_I}_{i}^*}+f_I^*Q^I(t), \qquad c^\text{SQ}(O,t) = \sum_{i \in O} \sqrt{\frac{\sigma_i^2\Delta}{E[Q_i(t)]}} \cdot Q_i(t)
\]
with $Q^I(t)=Q_0^I(t)+\sum_{\substack{i < j,\; i,j \in G}} Q_{ij}(t)$, the total realized trade volume in $I$ in quarter $t$. 
Similarly, the realized status quo cost in quarter $t$, $c^\text{SQ}(A,t)$ where $A$ is the set of all currencies. Throughout the paper, we sometimes drop the $t$
and use the notation $c(M)$ for both expected and realized costs, but specify which we refer to.

\section{Benchmark Greedy Algorithm}\label{app:greedy}

We compare the HAC algorithm with a greedy algorithm. This greedy algorithm scores each combination (group) of up to 6 non-vehicle currencies using the ratio $c^*/c^\text{SQ}$, where $c^*$ is the cost of trades within the group in a multi-asset pool and $c^\text{SQ}$ is the status quo cost of the same trades. The algorithm then selects groups in the order of descending scores while ensuring that every currency is in at most one selected group. The results using the empirical data described in Section~\ref{subsec:data} are the following pools (each also containing USD):

\smallskip

\noindent [CHF, CZK, EUR, GBP, RSD, SEK], [AUD, NZD], [JPY, KRW, SGD, THB], [CDF, TZS], [ALL, ISK, NOK, PLN], [BIF, BRL, CLP, COP, TRY], [KES, UGX], [PYG, UYU], [MZN, ZAR], [GHS, MXN, PHP], [CAD, MUR], [INR, PEN], [AMD, IDR], [DZD, JMD], [GTQ, NIO]

\smallskip

\noindent with a cost of 16{,}640 million USD and savings of 2{,}300 million USD, or 12\%. While the savings are on par with that of the HAC algorithm (13\% savings), the greedy algorithm takes over 4 hours on the same machine that runs the HAC algorithm in 1.6 seconds.

\fi

\end{document}